\documentclass[12pt]{article}
\usepackage{t1enc}
\usepackage[latin1]{inputenc}
\usepackage[english]{babel}

\usepackage{amsthm}
\usepackage{yfonts}

\usepackage{bbm}
\usepackage{bm}
\usepackage{mathrsfs}

\theoremstyle{plain}% default
\newtheorem{theorem}{Theorem}[section]
\newtheorem{lemma}[theorem]{Lemma}
\newtheorem{proposition}[theorem]{Proposition}
\newtheorem{corollary}[theorem]{Corollary}
\newtheorem{definition}[theorem]{Definition} 
\newtheorem{example}[theorem]{Example}  
\newtheorem{remark}[theorem]{Remark}

\usepackage[T1]{fontenc}              
\usepackage{amsmath,amssymb,bm,color,mathrsfs,mathtools}
\usepackage{amscd,verbatim}
\usepackage{graphicx}
\usepackage{subfigure}
\usepackage{multirow}
\usepackage{comment}
\usepackage[arrow,matrix,curve]{xy}
\usepackage{authblk}
\usepackage{tikz-cd}
\usepackage{colortbl}
\usepackage{rotating}
\usepackage{relsize}

\usepackage{tikz-cd}
\usetikzlibrary{decorations.markings}
\tikzset{->-/.style={decoration={
  markings,
  mark=at position #1 with {\arrow{>}}},postaction={decorate}}}
\tikzset{->-/.default=0.5}

\usepackage{pgfplots}

\pgfplotsset{compat=1.10}
\usepgfplotslibrary{fillbetween}
\usetikzlibrary{patterns}

\newcommand{\NN}{{\mathbb N}}
\newcommand{\RR}{{\mathbb R}}
\newcommand{\TT}{{\mathbb T}}
\newcommand{\ZZ}{{\mathbb Z}}

\newcommand{\cF}{{\mathcal F}}
\newcommand{\cK}{{\mathcal K}}

\newcommand{\cB}{{\mathcal B}}
\newcommand{\cS}{{\mathcal S}}

\newcommand{\cW}{{\mathcal W}}

\newcommand{\dd}{\mathrm{d}}

\newcommand{\sQ}{\mathscr{Q}}
\newcommand{\sC}{\mathscr{C}}

\newcommand{\Exp}{\mathrm{Exp}}

\begin{document}

\title{Cobordism invariance of topological edge-following states}

\author[1]{Matthias Ludewig}
\author[2]{Guo Chuan Thiang}
\affil[1,2]{School of Mathematical Sciences, University of Adelaide, Australia}
\affil[1]{Fakult\"at f\"ur Mathematik, Universit\"at Regensburg, Germany}
\affil[2]{Beijing International Center for Mathematical Research, Peking University, China}

%\date{today}

\maketitle

\begin{abstract}
We prove that a spectral gap-filling phenomenon occurs whenever a Hamiltonian operator encounters a coarse index obstruction upon compression to a domain with boundary. Furthermore, the gap-filling spectra contribute to quantised current channels, which follow and are localised at the possibly complicated boundary. This index obstruction is shown to be insensitive to deformations of the domain boundary, so the phenomenon is generic for magnetic Laplacians modelling quantum Hall systems and Chern topological insulators. A key construction is a quasi-equivariant version of Roe's algebra of locally compact finite propagation operators.  \end{abstract}

\section{Introduction}
One insight gained from the study of quantum Hall systems and topological insulators, is that their Hamiltonian operators $H_X$, acting on $L^2(X)$ for some manifold $X$ say, have spectral gaps that become filled up with ``topological boundary states'' when $H_X$ is compressed to an operator $H_W$ acting on some domain $W\subset X$ with boundary. Examples and rigorous proofs of such gap-filling phenomena are available in the case $X=\RR^2$ and $W$ a half-plane, e.g.\ \cite{Pule,BMR,FGW, KSB}. For general domains, not much is rigorously known about the fate of these boundary states. The physical expectation is that they persist due their ``topological origin'' and contribute to a quantised boundary-following current. Furthermore, to properly qualify as ``topological'' and fulfil their advertised novel applications, the boundary states should be robust against modifications of boundary conditions, see pp.\ 8 of \cite{AQPP} for a physical discussion.

\medskip

\noindent {\bf Outline.} In this paper, we use the tools of coarse geometry and $K$-theory to study the spectral gap-filling phenomenon in very general geometric settings, and especially its striking consequences in the form of boundary currents.

If $H_X$ is invariant with respect to a group action by $\Gamma$, the spectral projection $P_S$ onto a separated part $S$ of its spectrum defines an abstract $K$-theory class $[P_S]$ for the $\Gamma$-equivariant Roe algebra of $X$, denoted $C^*(X,\Gamma)$ (see Thm.\ \ref{thm:functional.calculus} and \cite{Ludewig-Thiang}). 
A subset $W\subset X$ is not generally invariant under $\Gamma$ (or even under any subgroup of $\Gamma$), but we nevertheless associate a Roe type algebra $Q^*(W,\Gamma)$ to it, which we call the \emph{quasi-equivariant} Roe algebra  introduced in Section \ref{sec:SES}. Here a quasi-equivariant operator on $L^2(W)$ has the crucial property that it eventually becomes equivariant as one moves sufficiently far away from $\partial W$. 
This ``periodization'' procedure maps $Q^*(W,\Gamma)$ onto $C^*(X,\Gamma)$, with kernel the Roe algebra of $W$ \emph{localised at} $\partial W$, denoted $C^*_W(\partial W)$; this gives the short exact sequence
\begin{equation*} %\label{ShortExactSequence}
\begin{tikzcd}
 0 \ar[r] & C^*_W(\partial W) \ar[r] & Q^*(W, \Gamma) \ar[r] &C^*(X, \Gamma) \ar[r] & 0.
\end{tikzcd}
\end{equation*}
%This latter algebra is a non-trivial extension of $C^*(X,\Gamma)$ by the Roe algebra of $W$ \emph{localised at} $\partial W$, denoted $C^*_W(\partial W)$. 
In Section \ref{sec:gap.filling.phenomenon}, we explain how the corresponding $K$-theoretic exponential map $\Exp_W:K_0(C^*(X,\Gamma))\rightarrow K_1(C^*_W(\partial W))$ applied to $[P_S]$ gives a ``boundary-localised'' obstruction for the compressed $H_W$ acting on $L^2(W)$ to maintain the spectral gaps adjacent to $S$. Thus ${\rm Exp}_W[P_S]\neq 0$ implies gap-filling when passing from $H_X$ to $H_W$ (see Fig.\ \ref{fig:spectral.projection}). 

Here, we observe that ${\rm Exp}_W[P_S]\in K_1(C^*_W(\partial W))$ is a \emph{non-equivariant} coarse index, so that \emph{no invariance property whatsoever is required of $\partial W$}. This is why our methods can address what is arguably the most astounding aspect of the gap-filling phenomenon by boundary states ---  it persists under deformations of the geometry of $\partial W$.

The direct computation of the obstruction ${\rm Exp}_W[P_S]$  might seem difficult except for special choices of $W$. In Section \ref{sec:cobordism}, we prove a certain \emph{cobordism invariance} of this obstruction, inspired by Roe's partitioned manifold index theorem \cite{Roe-partition, Roe-coarse-book}. We exploit this invariance to reduce the problem to standard half-spaces $\cW$ for which ${\rm Exp}_\cW[P_S]$ can be computed explicitly. For example, in Section \ref{sec:Euclidean.computation}, the obstruction is shown to be present for spectral projections of Chern insulators such as the magnetic Laplacian $H_{{\rm Lan},X}$ on $X=\RR^2$ (the \emph{Landau Hamiltonian} in physics), for \emph{generic} $W\subset \RR^2$, including $W$ with multiple boundary components. Thus we deduce a family of new results, that $H_{{\rm Lan},W}$ has no gaps in its spectrum (above the lowest Landau level), \emph{without having to solve the extremely difficult spectral problem for $H_{{\rm Lan},W}$}.

In Section \ref{sec:cyclic.cocycles}, under a polynomial growth condition on $\Gamma$, we prove Theorem \ref{thm:main.current.theorem} which provides a more concrete \emph{numerical} formula for the gap-filling indicator ${\rm Exp}_W[P_S]\in K_1(C^*_W(\partial W))$. As explained in Remark \ref{rem:physical.significance}, the numerical formula is physically the general expression for the current along the boundary $\partial W$ due to the gap-filling states of $H_W$, and there is no \emph{a priori} reason for it to take on only quantised values. In identifying this boundary current with a Fredholm index, Theorem \ref{thm:main.current.theorem} explains why it is quantised, invariant under ``coarse modifications'' of $\partial W$ and boundary conditions there, and invariant against perturbations of $H_X$ preserving the spectral separation of $S$.

\medskip

\noindent {\bf Related work:} For certain physical applications, e.g.\ quantum Hall effect, it is also important to establish robustness of the gap-filling spectra with respect to disorder and random potential terms, \cite{BES,KSB,FGW,PSB}. We do not address these issues, but rather focus on introducing new mathematical techniques to establish robustness \emph{with respect to choice of domain $W$}. 
We mention that the case of discrete $X$ (``tight-binding'' Hamiltonians describing lattice models) was studied recently by the second author \cite{Thiang-edge} using somewhat different techniques, and it provided preliminary evidence motivating this work. The authors also expanded the coarse index method to study gap-filling of Landau operators on the hyperbolic plane \cite{LT-coarse}. Two earlier works which introduced (uniform) Roe algebras and coarse geometry methods in topological phases of lattice models are \cite{Kubota, EwertMeyer}. Finally, a recent paper proposes coarse cohomology as invariants for interacting lattice systems \cite{Kapustin}, generalising the oft-utilised Chern classes in the non-interacting case.

\section{The quasi-equivariant Roe algebra}\label{sec:SES}

We begin with a rather general setting which allows the construction of what we call the \emph{quasi-equivariant short exact sequence} of $C^*$-algebras, Eq.~\eqref{ShortExactSequence}. This sequence can be considered a generalisation of the classical Toeplitz extension, reviewed in Example \ref{ex:Toeplitz}. 
After this section, we will adopt a more geometric setting, detailed at the beginning of \S \ref{sec:gap.filling.phenomenon}, which is suitable for spectral theory.

Let $(X, d, dx)$ be a proper metric measure space, i.e., one in which closed balls are compact. For two subsets $Y,Z\subset X$, their distance is denoted by \mbox{$d(Y,Z)={\rm inf}\{d(y,z):y\in Y, z\in Z\}$}. We also write \mbox{$B_R(A)=\{x\in X : d(x,A)\leq R\}$}. For $f\in L^\infty(X)$, we write $f$ for the corresponding multiplication operator on $L^2(X)$. We say that $A \in \cB(L^2(X))$ is {\em locally compact}, if $Af$ and $fA$ are compact for all compactly supported $f \in C_c(X)$. $A$ has {\em finite propagation} if there exists $R \geq 0$ such that $f Ag = 0$ whenever $f, g \in C_0(X)$ have supports at least $R$ apart. The closure in $\cB(L^2(X))$ of all locally compact, finite propagation operators is the {\em Roe algebra} $C^*(X)$ \cite{Roe-coarse-book}.

Assume moreover that $X$ carries a proper, isometric, measure-preserving action of a locally compact group $\Gamma$. This means that the group $\Gamma$ acts from the right on $L^2(X)$, via the unitary operators $U_\gamma, \gamma\in\Gamma$ defined by $U_\gamma w := \gamma^* w$, $w\in L^2(X)$. The {\em equivariant Roe algebra} $C^*(X, \Gamma)$ is the norm-closure in $\cB(L^2(X))$ of the $\Gamma$-invariant locally compact, finite propagation operators.

\vspace{1em}

\emph{Let $W \subset X$ be a closed subset, with $\partial W$ having zero measure}. Note that $W$ is not assumed to be preserved under $\Gamma$.

\vspace{1em}
We denote by $\Pi_W: L^2(X) \rightarrow L^2(W)$ the map that restrict functions to $W$ and by  $\Pi_W^*: L^2(W) \rightarrow L^2(X)$ the map that extends functions by zero to a function on $X$. For each $\gamma \in \Gamma$, we get the compressed operators 
\begin{equation*}
T_\gamma := \Pi_W U_\gamma \Pi_W^* \in \cB(L^2(W)). 
\end{equation*}
We remark that the resulting map $\Gamma \rightarrow \cB(L^2(W)), \gamma\mapsto T_\gamma$ is {\em not} a group homomorphism; in particular, the operators $T_\gamma$ are not generally invertible or even isometries.

\begin{definition}[The quasi-equivariant Roe algebra] \label{DefQuasiInvRoe}
~ We denote by $\sQ_0(W, \Gamma) \subseteq \cB(L^2(W))$ the algebra of all locally compact, finite propagation operators $A$ for which there exists $R\geq 0$ such that $(T_\gamma A - A T_\gamma)w = 0$ whenever both the support of $w\in C_c(W)$ and the support of $T_\gamma w$ have at least distance $R$ from $\partial W$. The {\em quasi-equivariant Roe algebra} $Q^*(W, \Gamma)$ is the closure of $\sQ_0(W, \Gamma)$ in the operator norm.
\end{definition}

Intuitively, when we are far away from $\partial W$, a quasi-equivariant operator on $L^2(W)$ behaves like a $\Gamma$-invariant one on $L^2(X)$. More precisely, we can relate the quasi-equivariant and equivariant Roe algebras under the following assumption on $W$:
\begin{equation} \label{ConditionOnW}
\begin{aligned}
&\text{\itshape For each $x\in X$, there exists a sequence $(\gamma_n)_{n \in \NN}$ in $\Gamma$}\\
&\text{\itshape such that $\gamma_n x \in W$ and $d(\gamma_n x,\partial W) \rightarrow \infty$.}
\end{aligned}
\end{equation}
For example, this is automatic if the action is cocompact and $d(\cdot,\partial W)$ is unbounded on $W$. Fig.\ \ref{fig:allowed.quasi.invariant} illustrates some subsets of the Euclidean plane (with standard $\Gamma=\ZZ^2$ action) satisfying/failing this criteria.

\begin{theorem}[The periodization map]
Under the assumption \eqref{ConditionOnW}, there exists a unique $*$-homomorphism
\begin{equation} \label{TheMapPi}
  \varpi: Q^*(W, \Gamma) \longrightarrow C^*(X, \Gamma), \qquad A \longmapsto \varpi A,
\end{equation}
the {\em periodization map}, with the property that for all $A \in \sQ_0(W, \Gamma)$, there exists $R \geq 0$ such that whenever the support of $w \in C_c(W)$ has distance at least $R$ from $\partial W$, then $\varpi A \Pi_W^* w = \Pi_W^* A w$.
\end{theorem}

\begin{proof}
To define $\varpi$, start with $A \in \sQ_0(W, \Gamma)$, and let $R\geq 0$ be such that $(T_\gamma A - A T_\gamma)w = 0$ whenever the supports of $w$ and $T_\gamma w$ have at least distance $R$ from $\partial W$. Now by assumption  \eqref{ConditionOnW} on $W$, for any compactly supported function $w \in C_c(X)$, we can find $\gamma \in \Gamma$ such that $U_\gamma w$ is supported in $W$, with distance at least $R+S$ from $\partial W$, where $S$ is the propagation speed of $A$. 
For such a  $\gamma$, the required properties of $\varpi$ implies that we must necessarily have
\begin{equation} \label{PreliminaryDefinition}
  \varpi A w = U_\gamma^*\varpi A U_\gamma w = U_\gamma^* \Pi_W^* A \Pi_W U_\gamma w,
\end{equation}
which we raise to a definition.
We show that this definition does not depend on the choice of $\gamma$. Indeed, if $\gamma^\prime$ is another element such that $U_\gamma w$ is supported in $W$, with distance at least $R+S$ from $\partial W$, we can write $\gamma^\prime = \gamma \eta$ and get 
\begin{equation*}
\begin{aligned}
 U_{\gamma^\prime}^* \Pi_W^* A \Pi_W U_{\gamma^\prime} w &= U_{\gamma}^* U_\eta^* \Pi_W^* A \Pi_W U_\eta U_{\gamma} w \\
 &= U_{\gamma}^* U_\eta^* \Pi_W^* A T_\eta  \Pi_W U_{\gamma} w \\ 
  &= U_{\gamma}^* U_\eta^* \Pi_W^* T_\eta A   \Pi_W U_{\gamma} w \\ 
    &= U_{\gamma}^* U_\eta^* U_\eta \Pi_W^*  A   \Pi_W U_{\gamma} w \\ 
        &= U_{\gamma}^*  \Pi_W^*  A   \Pi_W U_{\gamma} w \\ 
\end{aligned}
\end{equation*}
Here in the second step, we used that both $U_\eta U_\gamma w = U_{\gamma^\prime} w$ and $U_\gamma w$ have support in $W$, hence $\Pi_W U_\eta U_{\gamma} w = T_\eta  \Pi_W U_{\gamma} w$. Then, since the supports of $\Pi_W U_\gamma w$ and $T_\eta \Pi_W U_\gamma w = \Pi_W U_{\gamma^\prime} w$ have distance at least $R+S$ from $\partial W$ (by choice of $\gamma$, $\gamma^\prime$), we have $(A T_\eta - T_\eta A) \Pi_W U_\gamma w = 0$. Finally, we claim that $\Pi_W^* T_\eta A   \Pi_W U_{\gamma} w = U_\eta \Pi_W^*  A   \Pi_W U_{\gamma} w$, for which we have to show that both $\Pi_W^* A \Pi_W U_\gamma w$ and $U_\eta \Pi_W^* A \Pi_W U_\gamma w$ are supported in $W$. To see this, first notice that because $A$ has propagation speed at most $S$ and $U_\gamma w$ has support with distance at least $R+S$ from the boundary, the support of $A \Pi_W U_\gamma w$ still has distance at least $R$ from the boundary. Similarly, the support of  $U_\eta \Pi_W^* A \Pi_W U_{\gamma} w$ is contained in 
\begin{equation*}
\eta \cdot B_S\bigl(\mathrm{supp}(U_\gamma w)\bigr) = B_S \bigl(\eta \cdot \mathrm{supp}(U_\gamma w)\bigr) = B_S( \mathrm{supp}(U_{\gamma^\prime} w)),
\end{equation*} hence is contained in $W$, with distance at least $R$ from the boundary. This proves the claim and finishes the proof that Eq.\ \eqref{PreliminaryDefinition} is independent of the choice of $\gamma$ for all operators $A$ as above and all $w \in C_c(X)$.

Since $U_\gamma$, $\Pi_W^*$ and $\Pi_W$ have operator norm one, we have the estimate 
\begin{equation} \label{NormEstimatePi}
\|\varpi A w\|_{L^2(X)} = \|U_\gamma^* \Pi_W^* A \Pi_W U_\gamma w\|_{L^2(X)} \leq \|A\| \|w\|_{L^2(W)}
\end{equation}
for all $w \in C_c(X)$.
Therefore, as $C_c(X)$ is dense in $L^2(X)$, the operator $\varpi A$ defined by Eq.\ \eqref{PreliminaryDefinition} extends by continuity to a bounded operator on all of $L^2(X)$. Finally, we see that the estimate Eq.\ \eqref{NormEstimatePi} also implies that the map $\varpi: \sQ_0(W, \Gamma) \rightarrow \cB(L^2(X))$ is bounded, hence again extends by continuity to all of $Q^*(W, \Gamma)$.

We need to show that $\varpi A$ is $\Gamma$-equivariant. It suffices to verify this for $A \in \sQ_0(W, \Gamma)$ and $w \in C_c(X)$. Let $\eta \in \Gamma$ be arbitrary and $\gamma$ as in Eq.\ \eqref{PreliminaryDefinition}. Then $U_{\eta^{-1}\gamma} U_\eta w = U_\gamma w$ is supported in $W$, with distance at least $R+S$ from $\partial W$, hence
\begin{equation*}
  \varpi A U_\eta w = U_{\eta^{-1}\gamma}^* \Pi_W^* A \Pi_W U_{\eta^{-1}\gamma} U_\eta w = U_\eta U_\gamma^*  \Pi_W^* A \Pi_W U_{\gamma}  w = U_\eta \varpi A w,
\end{equation*}
as required.

Finally, we have to show that $\varpi$ is an $*$-homomorphism, which can again be verified on $\sQ_0(W, \Gamma)$. By the equivariance of $\varpi A$ just verified, it suffices to check that $\varpi A \varpi B w = \varpi(AB)w$ and $(\varpi A)^*w = \varpi A^*w$ for $w$ with support in $W$, far away from $\partial W$. However, this case is trivial.
\end{proof}

The map $\varpi$ admits a section,
\begin{equation}
\sigma: C^*(X, \Gamma) \rightarrow Q^*(W, \Gamma), \qquad A \longmapsto \Pi_W A \Pi_W^*,\label{eqn:section.map}
\end{equation}
in other words, we have $\varpi \sigma = \mathrm{id}$. This follows directly from the formula Eq.\ \eqref{PreliminaryDefinition} for $\varpi$. We emphasise that this map is {\em not} an algebra homomorphism, as it is not multiplicative. The existence of $\sigma$ in particular shows that the periodization map $\varpi$ is surjective, and its kernel turns out to be the localised Roe algebra at $\partial W$ (defined below), hence we obtain a short exact sequence of $C^*$-algebras,
\begin{equation} \label{ShortExactSequence}
\begin{tikzcd}
 0 \ar[r] & C^*_W(\partial W) \ar[r] & Q^*(W, \Gamma) \ar[r, "\varpi"] &C^*(X, \Gamma) \ar[l, bend left=20, "\sigma", dashed] \ar[r] & 0.
\end{tikzcd}
\end{equation}

\begin{definition}[\cite{HRY}, \S9 of \cite{Roe-coarse-book}]
  Let $\sC_{W,0}(\partial W) \subset \cB(L^2(W))$ be the subset of operators $A$ that are locally compact, of finite propagation and {\em supported near} $\partial W$, meaning that there exists $R \geq 0$ such that $Aw = 0$ for all $w \in C_c(W)$ the  support of which has distance at least $R$ from $\partial W$. The {\em Roe algebra of $W$ localised at $\partial W$}, denoted by $C^*_W(\partial W)$,  is the closure  of $\sC_{W,0}(\partial W)$ in the operator norm.
  \end{definition}

\begin{lemma} \label{LemmaLocalisedRoeAlgebraKernel}
  $C^*_W(\partial W)$ is the kernel of $\varpi$.
\end{lemma}

\begin{proof}
  It is clear from the definition Eq.\ \eqref{PreliminaryDefinition} that $\varpi A w = 0$ for all $A \in \sC_{W,0}(\partial W)$ and all $w \in C_c(X)$. By continuity, we also have $\varpi A w = 0$ for all $w \in L^2(X)$, hence $\sC_{W,0}(\partial W) \subseteq \ker(\varpi)$. Suppose, conversely, that $A \in \sQ_0(W, \Gamma) \cap \ker(\varpi)$. Then for all $w \in C_c(W)$, the support of which has distance at least $R+S$ from $\partial W$ (where $S$ is the propagation speed of $A$ and $R$ the constant from Def.~\ref{DefQuasiInvRoe}), we have 
  \begin{equation*}
    \Pi_W^* A w = \Pi_W \varpi A \Pi_W^* w = 0,
  \end{equation*}
  by formula Eq.\ \eqref{PreliminaryDefinition}, where we may choose $\gamma = 1$. Hence $A$ is supported near the boundary, so that $A \in \sC_{W,0}(\partial W)$.
  
  We have shown that $\sC_{W,0}(\partial W) = \ker(\varpi) \cap \sQ_0(W, \Gamma)$. Suppose now that $\varpi A = 0$ for a general $A \in Q^*(W, \Gamma)$, and let $A = \lim_n A_n$ with $A_n \in \sQ_0(W, \Gamma)$. Then $A^\prime_n := A_n - \sigma\varpi A_n \in \sQ_0(W, \Gamma)$ satisfies $\varpi A^\prime_n = 0$, hence $A^\prime_n \in \sC_{W,0}(\partial W)$, by our previous considerations. However, by continuity of $\varpi$ and $\sigma$, we have $\lim_n A^\prime_n = A - \sigma\varpi A = A$, hence $A$ is in the closure of $\sC_{W,0}(\partial W)$, which is $C^*_W(\partial W)$.
\end{proof}

\begin{example}[The Toeplitz extension]\label{ex:Toeplitz}
Consider $X = \RR$, with its canonical action of $\Gamma = \ZZ$. The interval $[0, 1]$ is a fundamental domain for the action, and we get \mbox{$C^*(X, \Gamma) \cong \mathcal{K}(L^2([0, 1])) \otimes C^*_r(\ZZ)$}, where $\mathcal{K}$ denotes the compact operators. The quasi-equivariant algebra for $W =\RR_+ \equiv [0, \infty)$ is $Q^*(W, \Gamma) = \mathcal{K}(L^2([0, 1])) \otimes \mathcal{T}$, where $\mathcal{T} = \{ \Pi_{\NN} A \Pi_{\NN}^* \mid A \in C^*_r(\ZZ)\} \subset \mathcal{B}(\ell^2(\NN))$ is the {\em Toeplitz algebra} obtained by compressing the reduced group $C^*$-algebra $C^*_r(\ZZ)\subset\mathcal{B}(\ell^2(\ZZ))$ to $\ell^2(\NN)$. The Roe algebra of $W$ localised at $\partial W$ is just $C^*_W(\partial W) = \cK(L^2(\RR_+)) \cong \cK(\ell^2(\NN)) \otimes \cK(L^2([0, 1]))$ in this case, and the quasi-equivariant short exact sequence Eq.\ \eqref{ShortExactSequence} is just the Toeplitz extension
\begin{equation*}
  \begin{tikzcd}
    0 \ar[r] & \cK\bigl(L^2(\NN)\bigr) \ar[r] & \mathcal{T} \ar[r] & C^*_r(\ZZ) \ar[r] & 0
  \end{tikzcd}
\end{equation*}
tensored with $\cK(L^2([0, 1]))$.
\end{example}

The short exact sequence Eq.\ \eqref{ShortExactSequence} yields the cyclic six-term exact sequence in $K$-theory
\begin{equation} \label{SixTermSequence}
\begin{tikzcd}
  K_0\bigl(C^*_W(\partial W)\bigr) \ar[r] & K_0\bigl( Q^*(W, \Gamma)\bigr) \ar[r, "\varpi_*"] & K_0\bigl(C^*(X, \Gamma)\bigr) \ar[d, "\mathrm{Exp}_W"]\\
  K_1\bigl(C^*(X, \Gamma)\bigr) \ar[u, "\mathrm{Ind}_W"] & K_1\bigl( Q^*(W, \Gamma)\bigr) \ar[l, "\varpi_*"] & \ar[l] K_1\bigl(C^*_W(\partial W)\bigr).
\end{tikzcd}
\end{equation}
For functorial operations, it is usual to assume that $L^2(X)$ is \emph{ample} (or \emph{adequate}, or \emph{standard}), i.e., the multiplication operator by $f\in C_0(X)$ is a compact operator in $\cB(L^2(X))$ only when $f=0$. This condition is always satisfied in the geometric setting of \S \ref{sec:gap.filling.phenomenon} onwards.

Notice that the algebra $C^*_W(\partial W)$ is the direct limit of its subalgebras of operators that are supported near the boundary, 
\begin{equation*}
C^*_W(\partial W) = \varinjlim C^*\bigl(B_R(\partial W) \cap W\bigr).
\end{equation*} 
Now since the inclusion map $\partial W \hookrightarrow B_R(\partial W) \cap W$ is a coarse equivalence for every $R \geq 0$, we have $C^*(B_R(\partial W) \cap W) \cong C^*(\partial W)$, the Roe algebra of $\partial W$ \cite[Thm.~2.7]{EwertMeyer}. While this isomorphism is non-canonical, one can choose it to be implemented by a unitary transformation of the underlying Hilbert space. Therefore, one obtains a {\em canonical isomorphism} of $K$-theory groups, cf.\ \S 5, Lemma 1 of \cite{HRY},
\begin{equation}
  K_*\bigl(C^*_W(\partial W)\bigr) \cong K_*\bigl(C^*(\partial W)\bigr). \label{eqn:localised.Roe.k.theory}
\end{equation}

One quick consequence of being able to ``thicken'' $\partial W$ is an invariance of the six-term sequence Eq.\ \eqref{SixTermSequence} under modifications of $\partial W$ within its thickening: 

\begin{proposition}
Let $W, W^\prime \subseteq X$ be two subsets satisfying the condition \eqref{ConditionOnW}, and assume that there exists $R \geq 0$ such that $\partial W^\prime \subset B_R(\partial W)$ and $\partial W \subseteq B_R(\partial W^\prime)$. Then there is a canonical isomorphism between the corresponding $K$-theory six-term sequences Eq.\ \eqref{SixTermSequence} that is the identity at the terms $K_i(C^*(X, \Gamma))$.
\end{proposition}

\begin{proof}
We may assume that $W^\prime \subseteq W$. Then the map $j: \cB(L^2(W^\prime)) \rightarrow \cB(L^2(W))$ given by sending $A$ to $\Pi_W^*A\Pi_W$ is an injective $*$-homomorphism, which sends $Q^*(W^\prime, \Gamma)$ to $Q^*(W, \Gamma)$ and $C^*_{W^\prime}(\partial W^\prime)$ to $C^*_W(\partial W)$. 

We claim that $j_*: K_i(C^*_{W^\prime}(\partial W^\prime)) \hookrightarrow K_i(C^*_W(\partial W))$ is an isomorphism. To this end, pick an open subset $V \subset W^\prime$ such that $\partial W^\prime, \partial W \subset B_R(V)$ for some $R\geq 0$ and such that $V \subset B_R(\partial W^\prime)$. By the choice of $V$, both $C^*_{W^\prime}(\partial W^\prime)$ and $C^*_W(\partial W)$ can be described as the closure of the space of locally compact, finite propagation operators on $W^\prime$ (respectively $W$) that are supported near $V$ instead of near $\partial W^\prime$ (respectively $\partial W$). By the considerations before, the inclusions $\iota_{W^\prime}: C^*(V) \hookrightarrow C^*_{W^\prime}(\partial W^\prime)$ and $\iota_{W}: C^*(V) \hookrightarrow C^*_W(\partial W)$ each induce isomorphisms in $K$-theory. On the other hand, we have $\iota_{W} = j \circ \iota_{W^\prime}$, hence $j$ must induce an isomorphism in $K$-theory as well.

The result now follows from the five lemma.
\end{proof}
For example, $W$ might be the standard half-plane $\RR_+\times\RR$ in the Euclidean plane, having straight boundary the vertical axis. Then we could modify $W\rightarrow W^\prime$ such that $\partial W^\prime$ remains within a vertical strip $[-R,R]\times\RR$ but is otherwise arbitrary. This encompasses the \emph{rough boundaries} considered in \cite{Prodan-rough} (see also \cite{Thiang-edge}) in the context of Euclidean lattice models. In \S \ref{sec:cobordism}, we will develop and exploit such ideas in greater generality.

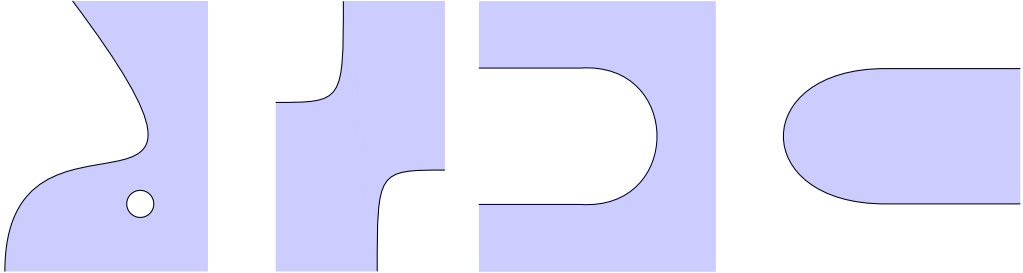
\begin{figure}
\begin{tikzpicture}[scale=0.9, every node/.style={scale=0.9}]

\draw[name path=A] (0,0) .. controls (0,3) and (4,0) .. (1,4);b
\draw[white, name path =B] (3,0) -- (3,4);
\tikzfillbetween[of=A and B]{blue, opacity=0.2}
\filldraw[white, draw=black] (2,1) circle (0.2cm);

\draw[name path=A] (4,2.5) .. controls (5,2.5) .. (5,4) ;
\draw[name path=B] (5.5,0) .. controls (5.5,1.5) .. (6.5,1.5);
\draw[white, name path=C] (4,0) -- (5.5,0);
\draw[white, name path=D] (5,4) -- (6.5,4);
\tikzfillbetween[of=D and B]{blue, opacity=0.2};
\tikzfillbetween[of=A and C]{blue, opacity=0.2};

\draw[thick, name path=A] (7,3) -- (8.5,3).. controls (10,3.1) and (10,0.9) ..(8.5,1) -- (7,1);
\draw[white, name path =B] (7,1) -- (7,3);
\filldraw[blue, opacity=0.2, draw=white] (7.01,0) -- (10.5,0) -- (10.5,4) -- (7.01,4);
\tikzfillbetween[of=A and B]{white};

\draw[name path=A] (15,3) -- (13,3).. controls (11,3) and (11,1) ..(13,1) -- (15,1);
\draw[white, name path =B] (15,1.5) -- (15,2.5);
\tikzfillbetween[of=A and B]{blue, opacity=0.2};

\end{tikzpicture}
\caption{The shaded domains in the first three diagrams show allowed half-spaces $W$ (in the sense of \S \ref{sec:gap.filling.phenomenon}) in the Euclidean plane $X=\RR^2$; they satisfy condition \eqref{ConditionOnW}. $W$ may be multiply-connected and/or have several unbounded boundary components. The last shaded figure fails condition \eqref{ConditionOnW}.}\label{fig:allowed.quasi.invariant}
\end{figure}

\section{Spectral gap filling phenomenon}\label{sec:gap.filling.phenomenon}

{\bf Conventions for the rest of the paper.} We will specialise to $X$ a complete connected Riemannian manifold with an effective, cocompact, properly discontinuous, isometric action of a discrete countable group $\Gamma$. Due to cocompactness of the action, the condition \eqref{ConditionOnW} on closed subsets $W \subset X$ (still with measure zero $\partial W$) becomes the condition that
\begin{equation} \label{ConditionOnW2}
  \text{\itshape{The function $d(x, X \setminus W)$ is unbounded.}}
\end{equation}
We will call such a subspace a \emph{half-subspace of $X$}, or simply a \emph{half-space}. 

\subsection{Functional calculus for Hamiltonians on a subspace}\label{sec:functional.calculus}

Let $H$ be a $\Gamma$-invariant elliptic differential operator on $X$, which is either of first order or of Laplace type (hence second order), with smooth coefficients.
We assume that $H$ is symmetric on the domain $C^\infty_c(X) \subset L^2(X)$; the general theory of such operators then asserts that it has a unique extension to an unbounded, self-adjoint operator on $L^2(X)$, which we denote by $H_X$. In the Laplace case, we assume additionally that $H_X$ is non-negative. More generally, we can consider $E$ a $\Gamma$-equivariant hermitian vector bundle over $X$, and a $\Gamma$-invariant Hamiltonian $H_X$ acting on its sections $L^2(X;E)$, but we will suppress the dependence on $E$ for ease of notation.

Let $W \subseteq X$ be a closed subset as above and consider $H$ on $C^\infty_c(W^\circ)$, where $W^\circ$ is the interior. We assume that we are given a self-adjoint extension $H_W$ of this operator determined by a local elliptic boundary condition. In the Laplace case, we assume that $H_W$ is still non-negative.

\begin{remark}
A typical example of such a boundary condition is the Dirichlet boundary condition $f|_{\partial W} = 0$, but there are usually many others (\cite{RS2} \S X). Let us mention that in the Laplace case, such self-adjoint extensions $H_W$ exist (see \cite{RS2} \S X.3), whereas $i\frac{d}{dx}$ on the half-line $W=[0,\infty)$ gives the classical first-order example with no self-adjoint extensions.
\end{remark}

\begin{theorem}\label{thm:functional.calculus} 
For each $\varphi \in C_0(\RR)$, we have $\varphi(H_X) \in C^*(X, \Gamma)$, $\varphi(H_W) \in Q^*(W, \Gamma)$, and $\varpi \varphi(H_W) = \varphi(H_X)$.
\end{theorem}

\begin{proof}
We first discuss the case where $H$ is a Laplace type operator. Since $H_X$ is positive, we have $\psi(\sqrt{H_X}) = \varphi(H_X)$ with $\psi(x) = \varphi(x^2)$. Since $\psi$ is an even function, we have the Fourier transform formula
\begin{equation*}
  \varphi(H_X)w = \psi\bigl(\sqrt{H_X}\bigr)w= \frac{1}{\pi} \int_0^\infty \hat{\psi}(s) \cos\bigl(s \sqrt{H_X}\bigr)w \, \dd s,
\end{equation*}
where $\hat{\psi}$ is the Fourier transform of $\psi$. The wave operator $\cos\bigl( s \sqrt{H_X}\bigr)$ has finite propagation, hence if $\hat{\psi}$ is compactly supported, $\varphi(H_X)$ has finite propagation as well. Local compactness is a consequence of elliptic regularity \cite[Prop.~10.5.1]{HigsonRoeBook}. Since functions with compactly supported Fourier transform are dense in $C_0(\RR)$, the result follows.

Let now $W$ be a subset as above. By assumption, $H_W$ is still positive, hence we get again
\begin{equation*}
  \varphi(H_W)w = \frac{1}{\pi} \int_0^\infty \hat{\psi}(s) \cos\bigl( s \sqrt{H_W}\bigr)w \dd s.
\end{equation*}
Now in $W$, both $\cos( s \sqrt{H_W})w$ and $\cos(s \sqrt{H_X})w$ solve the wave equation $(\partial_s^2 + H)w_s = 0$ in $W$ with initial conditions $w_0 = w$ and $\partial_s w_0 = 0$. By uniqueness of solutions to the wave equation, we have 
\begin{equation*}
  \cos\bigl(2 \pi s \sqrt{H_W}\bigr) w = \cos\bigl(2 \pi s \sqrt{H_X}\bigr)w
\end{equation*}
for all times $s \leq R$, whenever $w \in C_c(W)$ with $d(\mathrm{supp}(w), \partial W) >R$. In particular, we have
\begin{equation*}
\bigl(\varphi(H_W) - \sigma\varphi(H_X)\bigr)w = 2 \int_R^\infty \hat{\varphi}(s) \Bigl(\cos\bigl(2 \pi s \sqrt{H_W}\bigr) - \Pi_W\cos\bigl(2 \pi s \sqrt{H_X}\bigr)\Pi_W^*\Bigr)w\, \dd s,
\end{equation*}
where $\Pi_W: \cB(L^2(X)) \rightarrow \cB(L^2(W))$ denotes the restriction operator with adjoint $\Pi_W^*:\cB(L^2(W))\rightarrow \cB(L^2(X))$ the extension-by-zero operator as before. Now if $\hat{\varphi}$ has compactly supported Fourier transform, the right hand side vanishes if $R$ is large enough. This means that $\bigl(\varphi(H_W) - \sigma\varphi(H_X)\bigr)w = 0$ whenever $w$ has distance larger than $R$ from the boundary. Hence $\varphi(H_W) - \sigma\varphi(H_X)$ is supported near $\partial W$, and the result in this case follows from Eq.\ \eqref{ShortExactSequence}. For the general case, use the fact that functions $\varphi$ with compactly supported Fourier transform are dense in $C_0(\RR)$.

If $H_X$ is a first order operator, a similar Fourier transform argument can be made using the fact that the wave semigroups $e^{i s H_X}$ and $e^{i s H_W}$ have finite propagation speed.
\end{proof}

\subsection{Exponential map in $K$-theory detects spectral gap filling}\label{sec:gap.filling}
Suppose we are given a compact subset $S\subset {\rm Spec}(H_X)$, which is separated from the rest of the spectrum by spectral gaps, say $(a, \inf S)$, $(\sup S,b) \subseteq \RR \setminus \mathrm{Spec}(H_X)$. Let $P_S \in \cB(L^2(X))$ be the orthogonal projection onto the spectral subspace $L^2_S(X)$ determined by $S$. Owing to the existence of spectral gaps, $P_S$ can be written as a smooth, compactly supported function of $H_X$, hence by Thm.~\ref{thm:functional.calculus}, we have $P_S \in C^*(X, \Gamma)$ so that we obtain a class $[P_S] \in K_0(C^*(X,\Gamma))$; compare also \cite{Ludewig-Thiang}.

With $H_W$ as in \S\ref{sec:functional.calculus}, we note the following easy consequence of Theorem \ref{thm:functional.calculus}. 
\begin{corollary}\label{cor:specu.contains.specx}
The spectrum of $H_W$ contains that of $H_X$.
\end{corollary}
\begin{proof}
Suppose otherwise, that the resolvent $\rho(H_W)$ has some intersection with $\mathrm{Spec}(H_X)$. Pick a bounded open subinterval $V\subset \rho(H_W)$ such that $V\cap \mathrm{Spec}(H_X)\neq \emptyset$, and let $\tilde{\varphi}$ be a continuous bump function supported in $V$ which is not zero on $\mathrm{Spec}(H_X)$. This means that $\tilde{\varphi}(H_W)=0$ but $\tilde{\varphi}(H_X)\neq 0$. But Theorem \ref{thm:functional.calculus} would then give $0\neq \tilde{\varphi}(H_X)=\varpi\tilde{\varphi}(H_W)=\varpi(0)=0$.
\end{proof}
So in the passage from $H_X$ to $H_W$, a spectral gap of $H_X$ may become partially filled with new spectra of $H_W$. We are interested in whether the gap persists at all, or whether it instead gets completely filled. The following theorem gives a criterion for this.

\begin{theorem}[Spectral gap filling]\label{thm:gap.filling}
If ${\rm Exp}_W[P_S]\neq 0$ in the 6-term exact sequence Eq.\ \eqref{SixTermSequence} associated to $W\subset X$, then either $(a, \inf S)$ or $(\sup S,b)$ is in ${\rm Spec}(H_W)$.
\end{theorem}

\begin{proof}
Let $\varphi_S$ be any compactly supported smooth function such that $\varphi_S \equiv 1$ on $S$ and $\varphi_S = 0$ on $\mathrm{Spec}(H_X) \setminus S$. Then $P_S = \varphi_S(H_X)$. By Thm.~\ref{thm:functional.calculus}, we have $\varphi_S(H_X) \in C^*(X, \Gamma)$, $\varphi_S(H_W) \in Q^*(W, \Gamma)$ and $\varpi \varphi_S(H_W) = \varphi_S(H_X)$. Hence by definition of the exponential map, we have
\begin{equation*}
  \Exp_W([P_S]) = \bigl[\exp \bigl( - 2 \pi i \varphi_S(H_W)\bigr)\bigr].
\end{equation*}
Now suppose that there exist open sets $(c, d) \subseteq (a, \inf S)$, $(e, f) \subseteq (\sup S, b)$ not contained in $\mathrm{Spec}(H_W)$. Then we can choose the above function $\varphi_S$ in such a way that $\varphi_S \equiv 1$ on $(d, e)$ and $\varphi_S \equiv 0$ on $(-\infty, c] \cup [f, \infty)$. Since this function $\varphi_S$ takes the values zero and one on $\mathrm{Spec}(H_W)$, $\varphi_S(H_W)$ is a projection. Hence $\exp(-2\pi i \varphi_S(H_W))$ is the identity, hence $\Exp([P_S]) = 0$.

This shows that if $\Exp_W([P_S]) \neq 0$, the spectrum of $H_W$ cannot contain non-empty open subsets of both $[a, \inf S]$ and $[\sup S, b]$. Since $\mathrm{Spec}(H_W)$ is a closed set, this implies that one of the sets $[a, \inf S]$ or $[\sup S, b]$ must be contained in $\mathrm{Spec}(H_W)$.
\end{proof}

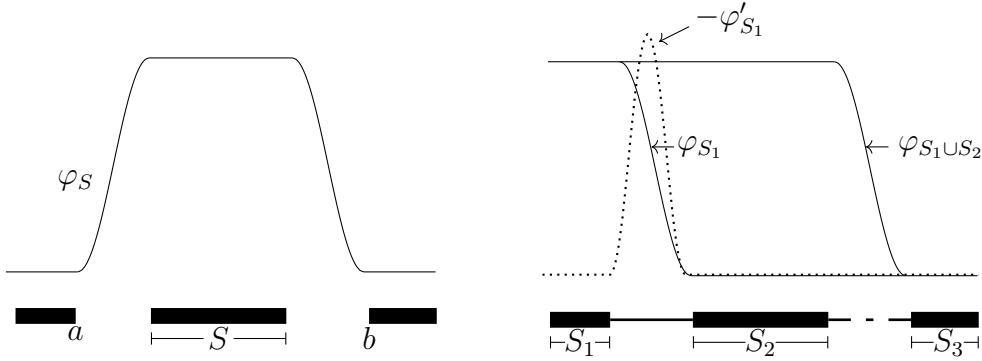
\begin{figure}[h]
\begin{minipage}{.45\linewidth}
\begin{tikzpicture}
%\draw[help lines] (0,0) grid (6,6);
       \begin{axis}[axis lines=none,yscale=0.6]
         \addplot[samples=20, smooth, domain=0.5:1]
           plot (\x, {0.5*(1+ sin(deg(2*pi*(\x-0.75))))});
           \addplot[samples=20, smooth, domain=2:2.5]
           plot (\x, {0.5*(1- sin(deg(2*pi*(\x-2.25))))});
            \addplot[samples=3, smooth, domain=0:0.5]
           plot (\x, {0});
            \addplot[samples=3, smooth, domain=1:2]
           plot (\x, {1});
           \addplot[samples=3, smooth, domain=2.5:3]
           plot (\x, {0});
      
       \end{axis}
            \draw[line width=6pt] (0.7,-0.3) -- (1.50,-0.3);
             \draw[line width=6pt] (2.5,-0.3) -- (4.3,-0.3);
              \draw[|-] (2.5,-0.6) -- (3.2,-0.6);
              \node at (3.4,-0.6) {$S$};
              \draw[-|] (3.6,-0.6) -- (4.3,-0.6);
              \draw[line width=6pt] (5.4,-0.3) -- (6.3,-0.3);
              \node at (1.5,1.5) {$\varphi_S$}; 
               \node[below] at (1.5,-0.3) {$a$}; 
                \node[below] at (5.4,-0.3) {$b$}; 
                \node at (3,3.8) {};
     \end{tikzpicture}
\end{minipage}
\hspace{.05\linewidth}
\begin{minipage}{.45\linewidth}
\begin{tikzpicture}
%\draw[help lines] (0,0) grid (6,6);
       \begin{axis}[axis lines=none,yscale=0.6]
         \addplot[samples=20, smooth, domain=0.5:1]
           plot (\x, {0.5*(1- sin(deg(2*pi*(\x-0.75))))});
           \addplot[samples=20, smooth, domain=2:2.5]
           plot (\x, {0.5*(1- sin(deg(2*pi*(\x-2.25))))});
            %\addplot[samples=3, smooth, domain=0:0.5]
           %plot (\x, {0});
            \addplot[samples=3, smooth, domain=0:2]
           plot (\x, {1});
           \addplot[samples=3, smooth, domain=1:3]
           plot (\x, {0});
      
       \end{axis}
       \draw[thick,dotted] (0.5,0.3) -- (1.4,0.3) .. controls (1.6,0.4) and (1.7,3.4) .. (1.9,3.5) .. controls (2.1,3.4) and (2.2,0.4) .. (2.4,0.3) -- (6.3,0.3);
       
            \draw[line width=6pt] (0.6,-0.3) -- (1.4,-0.3);
             \draw[line width=6pt] (2.5,-0.3) -- (4.3,-0.3);
             \draw[|-] (0.6,-0.6) -- (0.8,-0.6);
              \node at (1,-0.6) {$S_1$};
              \draw[-|] (1.2,-0.6) -- (1.4,-0.6);
              \draw[|-] (2.5,-0.6) -- (3.2,-0.6);
              \node at (3.4,-0.6) {$S_2$};
              \draw[-|] (3.6,-0.6) -- (4.3,-0.6);
              \draw[|-] (5.4,-0.6) -- (5.7,-0.6);
               \node at (5.9,-0.6) {$S_3$};
              \draw[-|] (6.1,-0.6) -- (6.3,-0.6);
              \draw[line width=6pt] (5.4,-0.3) -- (6.3,-0.3);
              \draw[line width=1pt] (1.4,-0.3) -- (2.5,-0.3);
               \draw[line width=1pt] (4.3,-0.3) -- (4.6,-0.3);
                \draw[line width=1pt] (4.8,-0.3) -- (4.9,-0.3);
                 \draw[line width=1pt] (5.1,-0.3) -- (5.4,-0.3);
              \node at (5.8,2) {$\varphi_{S_1\cup S_2}$}; 
              \node at (2.6,2) {$\varphi_{S_1}$}; 
               \node at (3,3.7) {$-\varphi_{S_1}^\prime$}; 
              \draw[->] (2.25,2)--(1.95,2); 
              \draw[->] (5.1,2)--(4.8,2); 
              \draw[->] (2.45,3.6)--(2.05,3.4); 
     \end{tikzpicture}
\end{minipage}
\caption{(L) Thick lines indicate the spectrum of $H_X$ as a subset of the real horizontal axis. A compact separated part $S$ of the spectrum has spectral projection $\varphi_S(H_X)$ for some smooth function $\varphi_S$ which is 1 on $S$ and $0$ elsewhere in the spectrum. (R) Suppose $H_X$ is bounded below, and $S_1, S_2, S_3$ are the first three separated parts of its spectrum, with ${\rm Exp}_W[\varphi_{S_1\cup S_2}(H_X)]=0$ but ${\rm Exp}_W[\varphi_{S_1}]\neq 0$. Then as indicated by the thinner horizontal lines, the spectrum of $H_W$ will include the entire gap between $S_1$ and $S_2$, but not necessarily the gap between $S_2$ and $S_3$. The dotted curve denotes $-\varphi_{S_1}^\prime$.}\label{fig:spectral.projection}
\end{figure}

\begin{remark}\label{rem:choice.of.phi}
If $H_W$ is bounded below, we may consider $S$ to be the spectrum of $H_X$ lying below some resolvent value (the \emph{Fermi level} is one example of physical interest). If  ${\rm Exp}_W([P_S])\neq 0$ so that gap-filling occurs, it must be the bounded spectral gap \emph{above} $S$ which is filled in the passage from $H_X$ to $H_W$, rather the unbounded gap below $S$. Another way to see this is to choose $\varphi_S$ to be 1 on $(-\infty, {\rm sup}\, S]$ and 0 on $[b,\infty)$, see Fig.\ \ref{fig:spectral.projection}.
\end{remark}

\section{Cobordism invariance of gap-filling}\label{sec:cobordism}

Guided by the constructions used in the partitioned manifold index theorem in \cite{Roe-coarse-book} \S4, we shall 
construct an index map $K_1(C^*_W(\partial W))\rightarrow \ZZ$ associated to a partition of $W$. To avoid confusion in what follows, we mention that in writing $\partial W$ for $\partial_X W$ (the boundary of $W$ in $X$), we had been keeping the background $X$ implicit by convention.

\vspace{1em}
\noindent {\bf Partitioning a space.} Given a subset $Z$ of a topological space $W$, the \emph{regular complement} is defined\footnote{Here, the symbol ${}^\circ$ denotes taking the interior, while $\overline{(\cdot)}$ denotes closure (both taken in $W$). } as $Z^\perp:=W\setminus Z^\circ=\overline{W\setminus Z}$, and it is easy to see that $(Z^\perp)^\perp=\overline{Z^\circ}$. We say that $Z$ is \emph{regular closed} if $Z=\overline{Z^\circ}$ ($=(Z^\perp)^\perp$). 

Let $W_+$ be a regular closed subset of $W$, then $W_-:=W_+^\perp$ is also regular closed. The interiors of $W_+$ and $W_-$ are disjoint, and the remaining subset
\begin{equation*}
N := W\setminus (W_+^\circ\cup W_-^\circ)=W_+^\perp\cap W_-^\perp=W_-\cap W_+,
\end{equation*}
is just their intersection. 
Note that
\begin{equation*}
\partial_W W_+=\overline{W_+}\cap\overline{W\setminus W_+}=\underbrace{W_+ \cap W_-}_{N}=\overline{W\setminus W_-}\cap\overline{W_-}=\partial_W W_-,
\end{equation*}
so that $N$ is simultaneously the boundary (inside $W$) of $W_+$ and of $W_-$. Thus, specifying a regular closed subset $W_+\subset W$ gives a sensible notion of partitioning $W$, and swapping $+$ and $-$ just switches the two ``sides'' of $N$.

\begin{example}\label{ex:standard.halfplane.example}
The standard example is $W$ the closed right-half Euclidean plane, and $W_+$ the closed upper-right quadrant. Then $W_-$ is the closed lower-right quadrant, while the partitioning set $N$ is the positive $x$-axis. Other examples are illustrated in Fig. \ref{fig:partition.modification}.
\end{example}

Since we are primarily interested in partitioning spaces $W$ that themselves arise as half-spaces inside $X$, we make the following restriction to avoid pathological partitions on $W$.

\begin{definition}[Admissible partition]\label{defn:admissible.subset}
Let $W\subset X$ be a half-subspace (as defined in \S \ref{sec:gap.filling.phenomenon}). An \emph{admissible} subset $W_+\subset W$ is a regular closed subset of $W$, with the following properties:
\begin{enumerate}
\item[(i)] For each $R\geq 0$, there exists $S\geq 0$ such that \mbox{$B_R(W_+)\cap B_R(W_-)\subset B_S(N)$}, or equivalently,
\mbox{$B_R(W_+)\cap B_R(W_-)\setminus B_S(N)=\emptyset$}. Here $N:=W_+\cap W_-$.
\item[(ii)] For each $R\geq 0$, the set $Q^{W;W_+}_R := B_R(X\setminus W)\cap B_R(N)$ is bounded.
\item[(iii)] $N$ has measure zero.
\end{enumerate}

\end{definition}
Note that $W_+$ is admissible iff its regular complement $W_-$ is admissible. 

\begin{remark}
If $W=X$, condition (i) is the precisely the notion of a \emph{coarsely excisive decomposition} of $W$ \cite{HRY}. Condition (ii) is a transversality condition between $N$ and $\partial W$ (the boundary of $W$ in $X$).
\end{remark}

\begin{example}\label{ex:nonadmissible.example}
Two examples of \emph{inadmissible} partitions of half-spaces $W$ in the Euclidean plane are illustrated below.

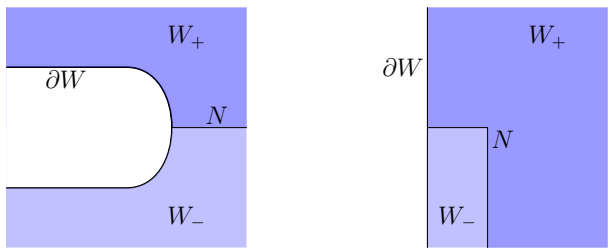
\begin{center}
\begin{tikzpicture}[scale=0.8, every node/.style={scale=0.8}]

\draw[thick, name path=A] (0,3) -- (2,3).. controls (3,3) and (3,1) ..(2,1) -- (0,1);
\draw[white, name path =B] (0,1) -- (0,3);
\filldraw[blue, opacity=0.4, draw=white] (0.01,0) -- (4,0) -- (4,4) -- (0.01,4);
\tikzfillbetween[of=A and B]{white};
\filldraw[white, opacity=0.4] (0,0)--(4,0)--(4,2)--(0,2);
\draw (2.75,2)--(4,2);
\draw (0,3) -- (2,3).. controls (3,3) and (3,1) ..(2,1) -- (0,1);
\node at (3,3.5) {$W_+$};
\node at (3,0.5) {$W_-$};
\node at (3.5,2.2) {$N$};
\node at (1,2.8) {$\partial W$};

\filldraw[blue, opacity=0.4, draw=white] (7,0) -- (10,0) -- (10,4) -- (7,4);
\filldraw[white, opacity=0.4, draw=white] (7,0) -- (8,0) -- (8,2) -- (7,2);
\draw (7,0)--(7,4);
\draw(7,2)--(8,2)--(8,0);
\node at (9,3.5) {$W_+$};
\node at (7.5,0.5) {$W_-$};
\node at (8.25,1.8) {$N$};
\node at (6.6,3) {$\partial W$};

\end{tikzpicture}
\end{center}
\end{example}

\begin{lemma}\label{lem:compact.commutator}
Let $\Pi_{W_+}$ be the multiplication operator on $L^2(W)$ by the characteristic function of $W_+$. For any admissible subset $W_+\subset W$ and any operator $A\in C^*_W(\partial W)$, the commutator $[\Pi_{W_+},A]$ is compact.
\end{lemma}

\begin{proof}
Let $A \in \cB(L^2(W))$ have finite propagation strictly less than $R>0$. Then for $f \in C_0(W)$ with support in $W_+$ of distance at least $R$ from $W_-$, $Af$ is still supported in $W_+$. Hence $[\Pi_{W_+}, A]f = \Pi_{W_+} Af - Af = 0$. Moreover, for any $g \in C_0(W)$, we have $f[\Pi_{W_+}, A]g = fAg - fA\Pi_{W_+}g = fA\Pi_{W_-} g$. However, $A \Pi_{W_-} g$ has support in $B_R(W_-)$, hence $fA\Pi_{W_-} g = 0$.

Similarly, if $f$ has support in $W_-$ with distance at least $R$ from $W_+$, then $Af$ is still supported in $W_-$ and $[\Pi_{W_+}, A]f = \Pi_{W_+}Af = 0$. Moreover, for any $g$, we have $f[\Pi_{W_+}, A] g = - fA\Pi_{W_+} g = 0$, since $A \Pi_{W_+} g$ is supported in $B_R(W_+)$. Together with the observations from the previous paragraph, this shows that $f[\Pi_{W_+}, A] g = 0$ whenever one of $f, g$ has support of distance at least $R$ from either $W_+$ or $W_-$. From the admissibility criterion (i) in Definition \ref{defn:admissible.subset}, $f[\Pi_{W_+}, A] g = 0$ whenever one of $f, g$ has support of distance at least $S$ from $N$.

Suppose additionally that $A$ is locally compact and such that $f Ag = 0$ whenever one of $f, g$ has support with at least distance $R'$ from $\partial W$. Now if one of $f, g$ has support at least distance $R'$ from $\partial W$, so has $\Pi_{W_+}f$ respectively $\Pi_{W_+}g$, hence $g[\Pi_{W_+}, A] f = 0$. Together with the argument before, this shows that $f[\Pi_{W_+}, A] g = 0$ whenever one of $f, g$ has support outside the relatively compact subset $Q^{W;W_+}_{{\rm max}\{R',S\}}$ in Definition \ref{defn:admissible.subset}. Choosing a compactly supported function $\chi \in C(W)$ with $\chi \equiv 1$ on $W\cap Q^{W;W_+}_{{\rm max}\{R',S\}}$, we therefore obtain that $[\Pi_{W_+}, A] = \chi[\Pi_{W_+}, A]\chi$, and the assumption that $A$ (hence also $[\Pi_{W_+}, A]$) is locally compact implies compactness of $[\Pi_{W_+}, A]$.

We have now proven the lemma for all operators $A \in \sC_{W,0}(\partial W)$. A general $A \in C^*_W(\partial W)$ can be written as $A = \lim_n A_n$, where $A_n \in \sC_{W,0}(\partial W)$ and the limit is in the operator norm. Therefore, $[\Pi_{W_+}, A] = \lim_n [\Pi_{W_+}, A_n]$ is a norm limit of compact operators, hence compact.
\end{proof}

It follows from Lemma \ref{lem:compact.commutator}, extended in the obvious way to direct sums, that for any invertible $A \in M_n(C^*_W(\partial W)^+)$, the compression $T_A := \Pi_{W_+} A \Pi_{W_+}^* \in \cB(L^2(W_+)^n)$ is invertible modulo compact operators, hence Fredholm (here $\Pi_{W_+}$ acts diagonally on $L^2(W)^n$). 
\begin{definition}\label{defn:index.hom}
Associated to an admissible subset $W_+$ of $W\subseteq X$ is the map
\begin{equation}
\theta_{W_+} \equiv\theta_{W; W_+}:K_1(C^*_W(\partial W))\rightarrow \ZZ,\qquad [u]\mapsto {\rm Index}\,T_u\label{eqn:Fredholm.index}
\end{equation}
where $u\in M_n(C^*_W(\partial W)^+)$ is a representative unitary.
\end{definition}
The extra subscript $W$ in $\theta_{W; W_+}$ will only be included when the role of $W$ needs to be emphasised. One easily checks that $\theta_{W_+}$ is well-defined and additive.

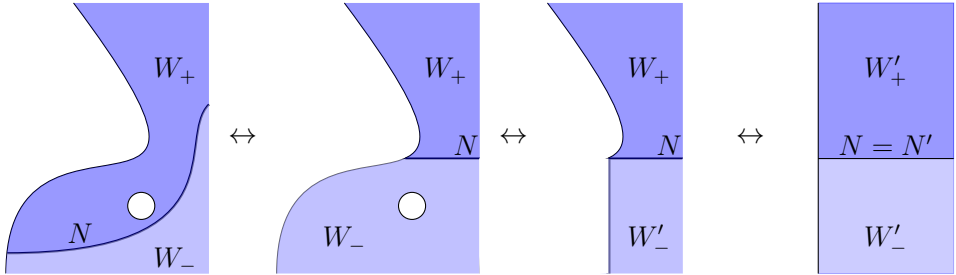
\begin{figure}
\begin{tikzpicture}[scale=0.9, every node/.style={scale=0.9}]

\draw[name path=A] (0,0) .. controls (0,3) and (4,0) .. (1,4);b
\draw[white, name path =B] (3,0) -- (3,4);
\draw[thick, name path=C] (0.02,0.3) .. controls (3.5,0.3) and (2.5,2) .. (3,2.5); 
\draw[white, name path=D] (0,0) -- (3,0);
\tikzfillbetween[of=A and B]{blue, opacity=0.4}
\tikzfillbetween[of=D and C]{white, opacity=0.4}
\filldraw[white, draw=black] (2,1) circle (0.2cm);
\node at (2.5,3) {$W_+$};
\node at (2.5,0.2) {$W_-$};
\node at (1.1,0.6) {$N$};

\node at (3.5,2) {$\leftrightarrow$};

\draw[name path=A] (4,0) .. controls (4,3) and (8,0) .. (5,4);
\draw[white, name path =B] (7,0) -- (7,4);
\draw[thick, name path=C] (5.9,1.7)--(7,1.7);
\draw[white, name path=D] (3.9,0)--(3.9,1.7);
\draw[white, name path =E] (7,0) -- (7,1.7);
\tikzfillbetween[of=A and B]{blue, opacity=0.4}
\tikzfillbetween[of=D and E]{white, opacity=0.4}
\filldraw[white, draw=black] (6,1) circle (0.2cm);
\node at (6.5,3) {$W_+$};
\node at (5,0.5) {$W_-$};
\node at (6.8,1.9) {$N$};

\node at (7.5,2) {$\leftrightarrow$};

\draw[name path=A] (7,0) .. controls (7,3) and (11,0) .. (8,4);
\draw[white, name path =B] (10,0) -- (10,4);
\draw[thick, name path=C] (8.9,1.7)--(10,1.7);
\draw[ultra thick, name path=E] (8.9,1.7)--(8.9,0);
\draw[white, name path=F] (7,1.7)--(7,-0.1);
\draw[white, name path =D] (10,0) -- (10,1.7);
\tikzfillbetween[of=A and B]{blue, opacity=0.4}
\tikzfillbetween[of=D and E]{white, opacity=0.4}
\tikzfillbetween[of=F and E]{white, opacity=1}
%\filldraw[white, draw=black] (6,1) circle (0.2cm);
\node at (9.5,3) {$W_+$};
\node at (9.5,0.5) {$W_-^\prime$};
\node at (9.8,1.9) {$N$};

\node at (11,2) {$\leftrightarrow$};

\filldraw[blue, opacity=0.4] (12,1.7)--(14,1.7)--(14,4)--(12,4);
\filldraw[blue, opacity=0.2] (12,1.7)--(14,1.7)--(14,0)--(12,0);
\draw (12,0)--(12,4);
\draw (12,1.7)--(14,1.7);
\node at (13,1.9) {$N=N^\prime$};
\node at (13,3) {$W_+^\prime$};
\node at (13,0.5) {$W_-^\prime$};

\end{tikzpicture}
\caption{The first two diagrams show two admissible partitions of the same non-simply-connected half-space $W$, which are bordant. The second to fourth diagrams illustrate Theorem \ref{thm:can.change.W} : while keeping $N$ fixed, we can modify $W_-\rightarrow W_-^\prime$ and then $W_+\rightarrow W_+^\prime$ to arrive at the standard partition of the half-plane $W^\prime$ in the last diagram, without changing the map $\theta_{W;W_+}\circ\Exp_W:K_0(C^*(X,\Gamma))\rightarrow\ZZ$.}\label{fig:partition.modification}
\end{figure}

We proceed to show that $\theta_{W_+}$ is somewhat insensitive to the choice of admissible $W_+\subset W$.

\begin{definition}\label{defn:bordant.partition}
Let $W_+$, $W_+^\prime$ be two admissible subsets of $W\subset X$, and let $W_+ \Delta W_+^\prime$ denote their {\em symmetric difference}, i.e.\ the set of $x \in W$ that are contained in exactly one of $W_+$, $W_+^\prime$. We say that $W_+$ and $W_+^\prime$ are \emph{bordant} if $W_+\cap W_+^\prime$ is another admissible subset of $W$, and the set $B_R(X\setminus W) \cap (W_+ \Delta W_+^\prime)$ is bounded for each $R\geq 0$. 
\end{definition}

Fig.\ \ref{fig:partition.modification} and Fig.\ \ref{fig:multiple.boundaries} show some examples of partitions by bordant and non-bordant admissible subsets.

\begin{proposition}\label{prop:cobordism.partition}
If $W_+$, $W_+^\prime$ are bordant admissible subsets of $W\subset X$, then $\theta_{W_+} = \theta_{W_+^\prime}$.
\end{proposition}
\begin{proof}
By a density argument, it suffices to check $\theta_{W_+}, \theta_{W_+^\prime}$ on a unitary $u=1+A\in C^*_W(\partial W)^+$ with $A\in \sC_{W,0}(\partial W)$. Write $L^2(W_+) = L^2(W_+\setminus W_+^\prime) \oplus L^2(W_+\cap W_+^\prime)$. With respect to this splitting, $T_u=\Pi_{W_+}u$ has the matrix representation
\begin{equation*}
  T_u = 
  \begin{pmatrix} 
  1+S & R_0 \\
  R_1  & \tilde{T}_u
  \end{pmatrix},\qquad \text{where} \quad \begin{cases} R_0 &= \Pi_{W_+\setminus W_+^\prime}A \Pi_{W_+\cap W_+^\prime}^*, \\ R_1 &= \Pi_{W_+\cap W_+^\prime}A \Pi_{W_+\setminus W_+^\prime}^*, \\ S &= \Pi_{W_+\setminus W_+^\prime}A \Pi_{W_+\setminus W_+^\prime}^*, \end{cases}
\end{equation*}
and $\tilde{T}_u=\Pi_{W_+\cap W_+^\prime}u\Pi_{W_+\cap W_+^\prime}^*$ is the Fredholm operator obtained by compressing $u$ to $W_+\cap W_+^\prime$. Since $B_R(\partial W) \cap (W_+\setminus W_+^\prime)\subset B_R(X\setminus W) \cap (W_+\Delta W_+^\prime)$ is bounded for any $R$ (thus relatively compact), and $A$ has finite propagation with support near $\partial W$, the operators $R_0, R_1, S$ have compact support. They are also locally compact, as $A$ is, hence compact.  Since the index is invariant under compact perturbations, we obtain
\begin{equation*}
  \mathrm{Index}(T_u) = \mathrm{Index}
  \begin{pmatrix} 
  1 & 0 \\
  0  & \tilde{T}_u
  \end{pmatrix} =  \mathrm{Index}(\tilde{T}_u).
  \end{equation*}
  Switching the roles of $W_+$ and $W_+^\prime$, for $T_u^\prime=\Pi_{W_+^\prime}u$, we also have
\begin{equation*}
\mathrm{Index}(T_u^\prime)=\mathrm{Index}(\tilde{T}_u).
\end{equation*}
Passing to direct sums, we obtain $\theta_{W_+}[u]=\theta_{W_+^\prime}[u]$ for any class in $K_1(C^*_W(\partial W))$.
\end{proof}

The remainder of this section is devoted to demonstrating that the map $\theta_{W;W_+}\circ\Exp_W:K_0(C^*(X,\Gamma))\rightarrow \ZZ$ depends only on the partitioning subset $N$ and so we can modify $W$ significantly (Theorem \ref{thm:can.change.W}).

\begin{proposition} \label{PropHalfCoincide}
Let $W, W^\prime \subset X$ be two half-spaces, and suppose $W_+ \subset W \cap W^\prime$ is admissible for both $W$ and $W^\prime$, and has the same boundary $N$ in $W$ and in $W^\prime$. Then the maps $\theta_{W; W_+}\circ\Exp_W$ and $\theta_{W^\prime; W_+}\circ\Exp_{W^\prime}$ coincide as homomorphisms $K_0(C^*(X, \Gamma))\rightarrow \ZZ$.
\end{proposition}

We will need the following lemma.

\begin{lemma}\label{lem:PropHalfCoincideLemma}
Let $W \subset X$ be closed, and let $Z_1, Z_2 \subset W$ be two subsets such that $B_R(Z_1) \cap B_R(Z_2)$ is bounded for each $R\geq 0$. Then $\Pi_{Z_1} A \Pi_{Z_2}$ is a compact operator on $L^2(W)$ for each element of $C^*(W)$.
\end{lemma}

\begin{proof}
If $A$ has finite propagation, then $\Pi_{Z_1} A\Pi_{Z_2}$ has bounded support by the assumption on $Z_1, Z_2$. It is moreover locally compact since $A$ is, hence compact. For general $A \in C^*(W)$, write $A$ as a norm limit over a sequence $(A_n)_{n \in \NN}$ of finite propagation operators. Then $\Pi_{Z_1} A \Pi_{Z_2}$ is the norm limit of the sequence $(\Pi_{Z_1} A_n\Pi_{Z_2})_{n\in \NN}$ of compact operators, hence compact. 
\end{proof}

\begin{proof}[Proof of Proposition \ref{PropHalfCoincide}]
\noindent {\em Step 1:} We show that $W_+$ remains admissible as a subset of $V=W \cup W^\prime$. Denote by $V_-$ the regular complement of $W_+$ in $V$, and notice that 
\begin{align*}
 V_- = \overline{(W \cup W^\prime) \setminus W_+} = \overline{(W \setminus W_+) \cup (W^\prime \setminus W_+)} &= \overline{(W \setminus W_+)} \cup \overline{(W^\prime \setminus W_+)} \\
 &= W_- \cup W_-^\prime.
\end{align*}
Here, $W, W^\prime, V\subset X$ are closed, so the above closures can be taken inside $X$. It follows that $W_+$ is regular closed in $V$, with boundary in $V$ being $W_+\cap V_-=(W_+\cap W_-)\cup(W_+\cap W_-^\prime)=N$. Thus condition (iii) is satisfied.

Condition (i) also follows: for any $R \geq 0$, there is a $S\geq 0$ such that
\begin{align*}
B_R(W_+)\cap B_R(V_-) &= \bigl(B_R(W_+)\cap B_R(W_-)\bigr)\cup\bigl(B_R(W_+)\cap B_R(W_-^\prime)\bigr)\\
&\subset B_S(N)
\end{align*}
by admissibility of $W_+$ in both $W$ and $W^\prime$.

Since $B_R(X\setminus V)\subset B_R(X\setminus W)\cup B_R(X\setminus W^\prime)$, taking intersection with $B_R(N)$ gives $Q^{V;W_+}_R\subset Q^{W;W_+}_R\cup Q^{W^\prime;W_+}_R$ which is bounded by assumption. So condition (ii) is satisfied.

\noindent {\em Step 2:} We show that $W_- \Delta W_-^\prime$ satisfy that $B_R(W_- \Delta W_-^\prime) \cap B_R(W_+)$ is bounded for each $R\geq 0 $. Given $R\geq 0$, pick $S\geq 0$ according to condition (i) to ensure that 
\begin{equation*}
B_R(W_-)\cap B_R(W_+)\setminus B_S(N)=\emptyset=B_R(W_-^\prime)\cap B_R(W_+)\setminus B_S(N).
\end{equation*}
Then $B_R(W_-\Delta W_-^\prime)\cap B_R(W_+)\setminus B_S(N)=\emptyset$ as well. Thus
\begin{align*}
B_R(W_- \Delta W_-^\prime)\cap B_R(W_+) &= B_R(W_- \Delta W_-^\prime)\cap B_R(W_+)\cap B_S(N) \\
&\subset B_R(W_- \Delta W_-^\prime)\cap B_S(N).
\end{align*}
Now note that a point in $W_-\Delta W_-^\prime$ is in $X\setminus W_+$, and either in $(X\setminus W_-)$ or $(X\setminus W_-^\prime)$, i.e.\
\begin{equation*}
W_-\Delta W_-^\prime\subset \bigl(X\setminus (W_+\cup W_-)\bigr)\cup \bigl(X\setminus (W_+\cup W_-^\prime)\bigr)=(X\setminus W)\cup (X\setminus W^\prime).
\end{equation*}
It follows that 
\begin{align*}
B_R(W_- \Delta W_-^\prime)\cap B_R(W_+)&\subset B_R\bigl((X\setminus W)\cup (X\setminus W^\prime)\bigr)\cap B_S(N)\\
&\subset \bigl(B_{R+S}(X\setminus W)\cap B_{R+S}(N)\bigr)\\
&\qquad\qquad \cup \bigl(B_{R+S}(X\setminus W^\prime)\cap B_{R+S}(N)\bigr)\\
&= Q_{R+S}^{W;W_+}\cup Q_{R+S}^{W^\prime;W_+},
\end{align*}
which is bounded by condition (ii).

\noindent {\em Step 3:} Let $p \in M_n(C^*(X, \Gamma))$ be a projection. Let $q = \sigma p \in M_n(Q^*(W, \Gamma))$, $q^\prime = \sigma^\prime p \in M_n(Q^*(W^\prime, \Gamma))$ be their canonical lifts by the section maps \eqref{eqn:section.map}. Extension by zero gives $L^2(W), L^2(W^\prime) \subset L^2(V)$, and we may view $q$, $q^\prime$ as elements of $C^*(V)$, the Roe algebra of $V$. We note that in passing from $W,W^\prime$ to $V$, the exponentials ${\rm exp}(-2\pi iq)$, ${\rm exp}(-2\pi iq^\prime)$ are merely modified by an identity operator on a complementary Hilbert space, so the indices of 
$T_{{\rm exp}(-2\pi i q)}$ and $T_{{\rm exp}(-2\pi i q^\prime)}$ are not affected. 

Let $Y_1 = W \cap W^\prime$ and $Y_2 = W_- \Delta W_-^\prime$ so that $V = Y_1 \cup Y_2$. Since $q$ and $q^\prime$ are obtained from the operator $p$, they coincide on $Y_1$. We therefore have
\begin{equation} \label{DecompositionDifference}
  q - q^\prime = \Pi_{Y_2} (q-q^\prime) \Pi_{Y_2} + \Pi_{Y_1}( q-q^\prime) \Pi_{Y_2} + \Pi_{Y_2} (q-q^\prime) \Pi_{Y_1}.
\end{equation}
Now write 
\begin{equation*}
  q^{k} - (q^\prime)^k = \sum_{j=1}^k (-1)^{j+1} q^{k-j} (q-q^\prime) (q^\prime)^{j-1}.
\end{equation*}
Using Eq.\ \eqref{DecompositionDifference}, we have 
\begin{equation*}
\begin{aligned}
  \Pi_{W_+} \bigl(q^{k} - (q^\prime)^k\bigr) \Pi_{W_+} = \sum_{j=1}^k (-1)^{j+1}& \Bigl(\underbrace{\Pi_{W_+}q^{k-j} \Pi_{Y_2}(q-q^\prime) \Pi_{Y_2}}_{\text{compact}}(q^\prime)^{j-1}\Pi_{W_+}\\
  &\!\!\!\!+ \underbrace{\Pi_{W_+}q^{k-j} \Pi_{Y_1}(q-q^\prime) \Pi_{Y_2}}_{\text{compact}}(q^\prime)^{j-1}\Pi_{W_+}\\
  &\!\!\!\!+\Pi_{W_+}q^{k-j} \underbrace{\Pi_{Y_2}(q-q^\prime) \Pi_{Y_1}(q^\prime)^{j-1}\Pi_{W_+}}_{\text{compact}}\Bigr),
\end{aligned}
\end{equation*}
where the indicated operators are compact as $W_+$, $Y_2\subset V$ satisfy the assumptions of Lemma \ref{lem:PropHalfCoincideLemma}, due to Step 2. We therefore get that the operator $\Pi_{W_+} \bigl(q^{k} - (q^\prime)^k\bigr) \Pi_{W_+}$ is compact for each $k$. In total,
\begin{equation*}
  T_{{\rm exp}(-2\pi i q)} - T_{{\rm exp}(-2\pi i q^\prime)} = \sum_{k=0}^\infty \frac{(-2\pi i)^k}{k!} \Pi_{W_+}(q^k - (q^\prime)^k)\Pi_{W_+},
\end{equation*}
where the sum converges in norm. Since each term in the sum is compact, the result is a compact operator. So $T_{{\rm exp}(-2\pi i q)}=\theta_{W;W_+}(\Exp_W[p])$ and $T_{{\rm exp}(-2\pi i q^\prime)}=\theta_{W;W_+^\prime}(\Exp_{W^\prime}[p])$ have the same indices.
 \end{proof}

\begin{lemma}\label{LemmaReverseOrientation}
Let $W_+\subset W$ be admissible, and $W_-$ its (admissible) regular complement. Then $\theta_{W_+} = - \theta_{W_-}$.
\end{lemma}

\begin{proof}
Let $u \in M_n(C^*_W(\partial W))^+$ represent a class $[u]$ in $K_1(C^*_W(\partial W))$. Since $\sC_{W,0}(\partial W) \subset C^*_W(\partial W)$ is dense, we may assume that $u$ has finite propagation and that $u-1$ is supported within finite distance of $\partial W$. Since $N$ has measure zero, we have the direct sum decomposition $L^2(W) = L^2(W_+) \oplus L^2(W_-)$. With respect to this,
\begin{equation*}
  u = \begin{pmatrix} \Pi_{W_+}u & K_1 \\ K_2 & \Pi_{W_-}u \end{pmatrix}.
\end{equation*}
Since $u$ has finite propagation, $K_1 = \Pi_{W_+} u \Pi_{W_-}$ and $K_2 = \Pi_{W_-} u \Pi_{W_+}$ have compact support as $u \in \sC_{W,0}(\partial W)$. Since $u-1$ is locally compact, $K_1$ and $K_2$ are compact. We obtain 
\begin{equation*}
\theta_{W_+}[u]=\mathrm{Index}(\Pi_{W_+} u) = - \mathrm{Index}(\Pi_{W_-} u)=-\theta_{W_-}[u],
\end{equation*}
 since $u$ is invertible.
\end{proof}

\begin{theorem}\label{thm:can.change.W}
Let $W_+$ and $W_+^\prime$ be admissible subsets of the half-spaces $W,W^\prime\subset X$ respectively, such that $N=W_+\cap W_-=W_+\cap W_-^\prime=W_+^\prime\cap W_-^\prime = W_+^\prime\cap W_-$. Suppose further that $W_+$ remains admissible for the modified half-space $W_1:=W_+\cup W_-^\prime$. Then 
\begin{equation*}
  \theta_{W;W_+}\circ{\rm Exp}_W=\theta_{W^\prime,W_+^\prime}\circ{\rm Exp}_{W^\prime}.
\end{equation*} 

\end{theorem}

\begin{proof}
With these assumptions, we may verify that $W_-^\prime$ is the regular complement of $W_+$ inside $W_1$, with common boundary $N=W_+\cap W_-^\prime$. So $W_+$ is admissible for both $W$ and $W_1$, while $W_-^\prime$ is admissible for both $W_1$ and $W^\prime$, with boundary being $N$ in all four cases. Using Lemma~\ref{LemmaReverseOrientation} and Proposition~\ref{PropHalfCoincide}, we deduce that
\begin{align*}
\theta_{W, W_+}\circ{\rm Exp}_W &= \theta_{W_1, W_+}\circ{\rm Exp}_{W_1}\\
&=-\theta_{W_1, W_-^\prime}\circ{\rm Exp}_{W_1}\\
&=-\theta_{W^\prime, W_-^\prime}\circ{\rm Exp}_{W^\prime}.
\end{align*}
The theorem follows.
\end{proof}

\section{Computations for $X$ the Euclidean plane}\label{sec:Euclidean.computation}
In this section, we study the Euclidean plane example, $X=\RR^2$, with $\Gamma=\ZZ^2$ the standard lattice of translations acting freely with fundamental domain $\cF=[0,1]\times[0,1]$.
The standard half-plane $\RR_+\times\RR$ is denoted $\cW$, and the standard quarter-plane $\RR_+\times\RR_+$ is denoted $\cW_+$.

In this case, the equivariant Roe algebra and reduced group $C^*$-algebra are related (see \S5.1.4 of \cite{Roe-coarse-book}),
\begin{equation*}
K_0\bigl(C^*(X,\Gamma)\bigr)\cong K_0\bigl(C^*_r(\Gamma)\otimes\cK\bigl(L^2(\cF)\bigr)\bigr)\cong K_0\bigl(C^*_r(\Gamma)\bigr)=K_0\bigl(C^*_r(\ZZ^2)\bigr),
\end{equation*}
where $\cK$ denotes the compact operators. Via a Fourier transform $C^*_r(\ZZ^2)\cong C(\TT^2)$ and Chern character map, it is easy to see that the RHS is $K_0(C^*_r(\ZZ^2))\cong \ZZ\oplus\ZZ$, where the two generators can be taken to be represented by the trivial projection and the Bott projection $\mathfrak{b}$. The Bott projection corresponds under the Serre--Swan theorem to a line bundle with first Chern class generating $H^2(\TT^2,\ZZ)\cong\ZZ$. Under the above isomorphism $K_0(C^*_r(\ZZ^2))\cong K_0(C^*(X,\Gamma))$, we will also think of $\mathfrak{b}$ as representing a generator of the latter.

\subsection{Coarse index and edge-travelling operator in $C^*_W(\partial W)$}

A fairly general class of half-spaces $W\subset X=\RR^2$ will have $\partial W$ coarsely equivalent, or even quasi-isometric to $\RR$, so that $K_1(C^*_W(\partial W))\cong K_1(C^*(\RR))$ in view of Eq.\ \eqref{eqn:localised.Roe.k.theory}. Because of this, it is instructive to recall and understand the result $K_1(C^*(\RR))\cong\ZZ$. 

\medskip

\noindent {\bf $K$-theory of the Roe algebra of the line.} It is known that $K_1(C^*(\RR))\cong\ZZ$ is generated by the so-called \emph{coarse index} ${\rm Ind}_c(D)$ of the Dirac operator $D= -i\frac{d}{dx}$ on $\RR$, see e.g.\ pp.\ 33 of \cite{Roe-coarse-book} and \cite{Roe-partition}. 
We give a more concrete \emph{hopping operator} $v\in (C^*(\RR))^+$ which also represents the generator of $K_1(C^*(\RR))$.

Pick any smooth $\psi\in L^2(\RR)$ which is supported in $[0,1]$, as illustrated in Fig.\ \ref{fig:hopping}. Then the translates $\gamma^*\psi, \gamma\in\ZZ$ provide an orthonormal basis for a copy of $\ell^2_{\rm reg}(\ZZ)\subset L^2(\RR)$. Let $v\in (C^*(\RR))^+$ be the unitary operator taking $\gamma^*\psi\mapsto (\gamma+1)^*\psi$ and acting as the identity on the orthogonal complement of $\ell^2_{\rm reg}(\ZZ)$ in $L^2(\RR)$.
Let $\Pi$ be the multiplication operator on $L^2(\RR)$ by the characteristic function on $\RR_+$ (the right half-line), which we use to compress an operator $A$ on $L^2(\RR)$ to an operator $T_A$ on $L^2(\RR_+)$. In much the same way that we took to construct Definition \ref{defn:index.hom}, there is a well-defined homomorphism (details can be found in pp.\ 28-29 of \cite{Roe-coarse-book}),
\begin{equation*}
\zeta:K_1(C^*(\RR))\rightarrow\ZZ,\qquad [u]\mapsto {\rm Index}(T_u).
\end{equation*}
The truncated hopping operator $T_v$ is essentially the unilateral right shift on $\ell^2(\NN)$ (direct summed with an identity operator), so its index is $-1$. So $\zeta$ is an isomorphism and $[v]$ indeed generates $K_1(C^*(\RR))$.

\begin{figure}[ht]
\begin{center}
\begin{tikzpicture}[scale=0.8, every node/.style={scale=0.8}]
\node at (0,-2) {};
\draw[thick,->] (-2.5,0) -- (8.5,0) node[anchor=north west] {};
\foreach \x in {-1,0,1,2,3,4}
   \draw (2*\x cm,1pt) -- (2*\x cm,-1pt) node[anchor=north] {$\x$};
   \draw[dotted] (0,0)--(0,1.6) {};
\foreach \y in {-1,0,1,2,3}   
	\draw (2*\y+0.2,0) .. controls (2*\y+0.6,0.1) and (2*\y+0.8,1.4) .. (2*\y+1,1.5) .. controls (2*\y+1.2,1.4) and (2*\y+1.4,0.1) .. (2*\y+1.8,0) ;
\node at (1.5,1) {$\psi$};
\node (a) at (-1,1.6) {};
\node (b) at (1,1.6) {};
\node (c) at (3,1.6) {};
\node (d) at (5,1.6) {};
\node (e) at (7,1.6) {};
\path
    (a) edge[bend left, ->-] node [right] {} (b)
    (b) edge[bend left, ->-] node [right] {} (c)
    (c) edge[bend left, ->-] node [right] {} (d)
    (d) edge[bend left, ->-] node [right] {} (e);
\node at (6.5,2) {$v$};
\end{tikzpicture}
\hspace{1em}
\begin{tikzpicture}[scale=0.7, every node/.style={scale=0.7}]
\draw[step=2cm,gray,dotted] (-2.1,-2.1) grid (6,6);
\filldraw[opacity=0.05] (0,-2.1)--(6.1,-2.1)--(6.1,0)--(0,0);
\filldraw[opacity=0.1] (0,0)--(6.1,0)--(6.1,6.1)--(0,6.1);
\draw[thick,->] (-2.2,0) -- (6.5,0) node[anchor=north west] {};
\foreach \x in {-1,1,2,3}
   \draw (2*\x cm,1pt) -- (2*\x cm,-1pt) node[anchor=north] {$\x$};
\draw[thick,->] (0,-2.1) -- (0,6.3) node[anchor=north west] {};
   \foreach \z in {-1,1,2,3}
   \draw (1pt,2*\z cm) -- (-1pt,2*\z cm) node[anchor=east] {$\z$};
   \draw[dotted] (0,0)--(0,1.6) {};
\node at (-0.3,-0.3) {$0$};
\foreach \y in {-1,0,1,2}   
	\draw (0.2,2*\y+0.2) .. controls (0.6,2*\y+0.3) and (0.8,2*\y+1.6) .. (1,2*\y+1.7) .. controls (1.2,2*\y+1.6) and (1.4,2*\y+0.3) .. (+1.8,2*\y+0.2) ;
\node at (0.5,1) {$\phi$};
\node (a) at (1.8,7) {};
\node (b) at (1.8,5) {};
\node (c) at (1.8,3) {};
\node (d) at (1.8,1) {};
\node (e) at (1.8,-1) {};
\path
    (a) edge[bend left, ->-] node [right] {} (b)
    (b) edge[bend left, ->-] node [right] {} (c)
    (c) edge[bend left, ->-] node [right] {} (d)
    (d) edge[bend left, ->-] node [right] {} (e);
\node at (2.3,5.8) {$w$};
\end{tikzpicture}
\end{center}
\caption{(L) Hopping operator on a line. (R) Edge-travelling operator on a half-plane.}\label{fig:hopping}
\end{figure}

\medskip

\noindent {\bf Edge-travelling operator on standard half-plane $\cW$.}
With $\cW=\RR_+\times\RR$, so that $\partial \cW=\{0\}\times \RR$, the localisation principle Eq.\ \eqref{eqn:localised.Roe.k.theory} gives 
\begin{equation*}
\ZZ\cong K_1(C^*(\partial \cW))\cong K_1(C^*([0,1]\times\partial \cW))\cong K_1(C^*_\cW(\partial \cW)).
\end{equation*}
A construction similar to that of $v$ above, will therefore yield a representative generator of $K_1(C^*_\cW(\partial \cW))\cong\ZZ$ (see Fig.\ \ref{fig:hopping}). Namely, pick a smooth $\phi\in L^2(X)$ supported in $[0,1]\times[0,1]$, so that its translates by $\gamma\in \ZZ^2$ (resp.\ $\gamma\in \NN\times\ZZ$) provide an orthonormal basis for a copy of $\ell^2_{\rm reg}(\ZZ^2)$ inside $L^2(X)$ (resp.\ $\ell^2(\NN\times\ZZ)$ inside $L^2(\cW)$). On $L^2(\cW)$, let $w$ denote the ``edge-travelling operator'' which acts on the ``edge-subspace'' $\ell^2(\{0\}\times\ZZ)$ by downward translation $(0,n)^*\phi\mapsto (0,n-1)^*\phi$, and is the identity operator on the orthogonal complement. Then $w$ is a unitary operator in $(C^*_\cW(\partial\cW))^+$ representing a generator of $K_1(C^*_\cW(\partial\cW))\cong\ZZ$.

\begin{proposition}\label{prop:surjective.exponential}
With $\cW$ the standard half-plane in $X=\RR^2$, the $K$-theory exponential map for
\begin{equation}
0\rightarrow C^*_\cW(\partial \cW)\rightarrow Q^*(\cW,\Gamma)\rightarrow C^*(X,\Gamma)\rightarrow 0\label{eqn:half-plane.SES}
\end{equation}
is surjective. It maps the Bott projection class $[\mathfrak{b}]\in K_0(C^*(X,\Gamma))$ to a generator of $K_1(C^*_\cW(\partial\cW))\cong\ZZ$ (the class of the edge-travelling operator $w$ described above). 
\end{proposition}
\begin{proof}
The discrete version of Eq.\ \eqref{eqn:half-plane.SES} is
\begin{equation}
0\rightarrow \cK(\ell^2(\NN))\otimes C^*_r(\ZZ)\rightarrow C^*_r(\NN\times\ZZ)\rightarrow C^*_r(\ZZ^2)\rightarrow 0\label{eqn:semigroup.Toeplitz}
\end{equation}
where $C^*_r(\NN\times\ZZ)$ is the (reduced) semigroup $C^*$-algebra for $\NN\times \ZZ$. Let $w|$ be the restriction of $w$ to $\ell^2(\NN\times\ZZ)\subset L^2(\cW)$. Then $w|$ is a unitary operator in $(\cK(\ell^2(\NN))\otimes C^*_r(\ZZ))^+\subset \cB(\ell^2(\NN\times\ZZ))$ which effects ``downward translation along the first column'' and does nothing elsewhere. It is clear that $w|$ represents the generator of $K_1(\cK\otimes C^*_r(\ZZ^2))\cong \ZZ$. However, when $w|$ is regarded as an element in the larger algebra $C^*_r(\NN\times\ZZ)\subset \cB(\ell^2(\NN\times\ZZ))$, a combination of the Toeplitz extension's (Example \ref{ex:Toeplitz}) $K$-theory sequence and the K\"{u}nneth theorem shows that $[w|]$ trivialises in $K_1(C^*_r(\NN\times\ZZ))$, see \S 2.2.3 of \cite{Thiang-edge}. Since $w\in C^*_\cW(\partial\cW)$ is just the extension of $w|$ by the identity operator on an orthogonal subspace, it also represents the trivial class when regarded as an element of the larger algebra $Q^*(\cW,\Gamma)$. Exactness of the long exact sequence for Eq.\ \eqref{eqn:half-plane.SES} means that ${\rm Exp}_\cW:K_0(C^*(X,\Gamma))\rightarrow K_1(C^*_\cW(\partial\cW))$ must be surjective. The $K$-theory exponential map is trivial on identity projections $[1]$, so we must have ${\rm Exp}_\cW([\mathfrak{b}])=[w]$ (up to a sign).
\end{proof}

\begin{remark}
The term \emph{edge-travelling operator} was coined in \cite{Thiang-edge}, in an investigation of gap-filling by ``edge-following topological states'' in lattice models of so-called \emph{Chern insulators} arising in physics. Eq.\ \eqref{eqn:semigroup.Toeplitz} is an example of a \emph{semigroup Toeplitz extension}, for the case $\NN\times\ZZ\subset\ZZ^2$.
\end{remark}

\subsection{Application to Chern insulators and Landau Hamiltonian}
Quite generally, a \emph{Chern insulator} can be defined as a Hamiltonian $H_X=H_{{\rm Chern},X}$, which has some spectral projection $P_S=\varphi_S(H_{{\rm Chern},X})$ having $K$-theory class $k[1]\oplus  j[\mathfrak{b}]\in K_0(C^*(X,\Gamma))$ with $j\neq 0$. Usually, $S$ is taken to be the subset of the spectrum below some prescribed \emph{Fermi energy} $E_F\not\in {\rm Spec}(H_{{\rm Chern},X})$. Such a ``topological projection'' $P_S$ is said to have \emph{Chern number $j$}.

The non-vanishing abstract homotopy invariant of $P_S$ has dramatic consequences. First, taking the standard half-plane $\cW$ as domain, Proposition \ref{prop:surjective.exponential} says that $\Exp_\cW[P_S]=j[w]\neq 0$. Upon passing to $H_{\rm Chern,\cW}$, Theorem \ref{thm:gap.filling} guarantees that the spectral gap above $S$ is completely filled up. In particular, the Fermi energy $E_F$ is in the spectrum of $H_{\rm Chern,\cW}$; physicists call this property of $H_{\rm Chern,\cW}$ \emph{gaplessness} (at $E_F$).

\begin{example}
Let $A= x\,dy$ be a connection 1-form on $X=\RR^2$ with curvature $ dx\wedge dy$, corresponding to a uniform magnetic field applied perpendicularly to the plane, with unit flux per unit area.
The magnetic Laplacian, or \emph{Landau Hamiltonian},
\begin{equation*}
H_{{\rm Lan},X}=\frac{1}{2}(d-iA)^*(d-iA)
\end{equation*} 
is self-adjoint and has the harmonic oscillator spectrum $\frac{1}{2}+\mathbb{N}_0$, with each eigenvalue (\emph{Landau level}) being infinitely degenerate.

For each $\gamma=(a,b)\in\RR^2$, define the \emph{magnetic translations} $U_\gamma$ on $L^2(\RR^2)$ by $(U_{\gamma}f)(x,y)=f(x-a,y-b)e^{iay}$, then $H_{{\rm Lan},X}$ commutes with each $U_{\gamma}$. While $\gamma\mapsto U_\gamma$ only gives a \emph{projective} unitary representation of the group $\RR^2$, when we restrict to those $\gamma$ in the lattice $\Gamma=(\sqrt{2\pi}\ZZ)^2\subset\RR^2$, we do get a genuine unitary representation of $\Gamma\cong\ZZ^2$ on $L^2(\RR^2)$, and $H_{{\rm Lan},X}$ is of course $\Gamma$-invariant.
%Observe that $H_{{\rm Lan},X}$ is invariant under the standard lattice $\Gamma=\ZZ^2$ of translation operators on $L^2(\RR^2)$. 

For each $j\geq 1$, consider the spectral projection onto the first $j$ Landau levels $\frac{1}{2}, \ldots, \frac{2j-1}{2}$, which can be written as $\varphi_j(H_{{\rm Lan},X})\in C^*(X,\Gamma)\cong C^*_r(\Gamma)\otimes\cK$ for a suitable function $\varphi_j$. This projection defines an element 
\begin{equation*}
\bigl[\varphi_j(H_{{\rm Lan},X})\bigr]\in K_0\bigl(C^*(\RR^2,\ZZ^2)\bigr)\cong K_0\bigl(C^*_r(\ZZ^2)\bigr)\cong K^0(\TT^2).
\end{equation*}
It is known that after taking the Chern character, $[\varphi_j(H_{{\rm Lan},X})]$ has Chern class being $j$ times the generator of $H^2(\TT^2,\ZZ)$, see \cite{BES} Lemma 5, \cite{Kunz} Eq.\ 3.55, \cite{DeNittis} \S3.7. In other words, the $H_{\rm Lan,X}$ is a Chern insulator and the projection $\varphi_j(H_{{\rm Lan},X})$ has Chern number $j$.

We deduce that each spectral gap $(\frac{2j-1}{2}, \frac{2j+1}{2})$  of $H_{{\rm Lan},X}$ will be filled up with new spectra of $H_{{\rm Lan},\cW}$. This deduction is corroborated by an exact calculation of the spectrum of the half-plane Dirichlet $H_{{\rm Lan},\cW}$ as an unbroken half-line $[\frac{1}{2},\infty)$, see \cite{Pule}. The half-plane  Neumann Laplacian $H_{{\rm Lan},\cW}$ has similar features, see \cite{BMR} and references therein.

\end{example}

The utility of Theorem \ref{thm:can.change.W} is that we can proceed to deduce the same gap-filling phenomenon for $H_{{\rm Chern},W}$ on \emph{generic} domains $W\subset X=\RR^2$, \emph{without having to solve the extremely difficult spectral problem for $H_{{\rm Chern},W}$}! Even if we modify $\cW$ quite drastically into another half-space domain $W$ (within the assumptions of Theorem \ref{thm:can.change.W}) we still have, for a spectral projection $P_S=\varphi_S(H_{{\rm Chern},X})$ with Chern number $j\neq 0$, that
\begin{align}
\theta_{W;W_+}({\rm Exp}_W[P_S])&=\theta_{\cW;\cW_+}({\rm Exp}_\cW[P_S])\nonumber\\
&=\theta_{\cW;\cW_+}(j\cdot [w])\nonumber\\
&=j\cdot{\rm Index}\,(T_w)=j\cdot{\rm Index}\,({\rm Shift}) =j\neq 0\label{eqn:edge.index.computation}.
\end{align}
In the last line, we used the observation that for the edge-travelling operator $w$, its compression $T_w$ to the upper-right quadrant is essentially the unilateral downward-shift operator on $\ell^2(\{0\}\times \NN)$ which has Fredholm index 1. Then ${\rm Exp}_W[\varphi_S(H_{{\rm Chern},X})]\neq 0$ and Theorem \ref{thm:gap.filling} implies filling of the spectral gap above $S$ when passing to $H_{{\rm Chern},W}$.

\begin{remark}
In particular, for the Landau Hamiltonian, this means that there are no gaps in the spectrum of $H_{{\rm Lan},W}$ above the lowest Landau level $\frac{1}{2}$.
\end{remark}

\begin{remark}
We point out that the existence of extended but boundary-localised nature of the new states of $H_W$ that appear somewhere in spectral gaps of $H_X$, was deduced in \cite{FGW} for a large class of magnetic Laplace-type operators and fairly general domains $W\subset\RR^2$.
\end{remark}

\subsection{Domains with multiple boundary components}

Consider a half-space $W$ with $\partial W=Y_1\amalg Y_2$ comprising the two components of a hyperbola. To compute $K_1(C^*_W(\partial W))\cong K_1(C^*(\partial W))$, we note that $\partial W$ is coarsely equivalent to the cross $\mathlarger{\mathlarger{+}}$ formed by the asymptotes. Furthermore, $\mathlarger{\lrcorner}$ and $\mathlarger{\ulcorner}$ (which are each quasi-isometric to a Euclidean line) form a coarsely excisive decomposition \cite{HRY} of the cross with intersection a single point. It follows from the coarse Mayer--Vietoris sequence (\S5 of \cite{HRY}) that
\begin{align*}
K_1(C^*_W(\partial W)) &\cong K_1(C^*(\partial W))\\
&\cong K_1(C^*(\lrcorner))\oplus K_1(C^*(\ulcorner))\oplus K_0(C^*({\rm pt}))\cong\ZZ^3.
\end{align*}
As verified below, representative generators for $K_1(C^*_W(\partial W))\cong\ZZ^3$ can be taken to be the edge-travelling operators $w_{Y_1}, w_{Y_2}$ hopping along the boundary components $Y_1$ and $Y_2$ respectively, together with the operator $w_Z$ hopping rightwards along the horizontal asymptote.

\begin{figure}
\begin{center}
\begin{tikzpicture}[scale=0.8, every node/.style={scale=0.9}]

\draw (2,5)--(2,9);
\draw (2,7)--(0.5,7);
\filldraw[blue,opacity=0.4] (0.5,7)--(2,7)--(2,9)--(0.5,9);
\filldraw[blue,opacity=0.2] (0.5,7)--(2,7)--(2,5)--(0.5,5);
\draw[->-] (2,6)--(2,6.1) {};

\node at (1,4.5) {$\updownarrow$};

\draw[name path=A] (0,2.5) .. controls (1,2.5) .. (1,4) ;
\draw[name path=B] (1.5,0) .. controls (1.5,1.5) .. (2.5,1.5);
\draw[white, name path=C] (0,0) -- (1.5,0);
\draw[white, name path=D] (1,4) -- (2.5,4);
\draw[thick, name path=E ] (0,0.8)--(1.6,1.2);
\tikzfillbetween[of=D and B]{blue, opacity=0.4};
\tikzfillbetween[of=A and C]{blue, opacity=0.4};
\filldraw[white, opacity=0.4] (0,0)--(2.5,0)--(2.5,1.2)--(1.6,1.2)--(0,0.8)--(0,0);
\node at (1,1.35) {$N_1$};
\node at (0.5,2.8) {$Y_1$};
\node at (2,1.8) {$Y_2$};
\draw[->-] (1.5,0.5)--(1.5,0.6) {};
\draw[->-] (0.95,3)--(0.95,2.9) {};

\node at (3.5,2) {$\not\leftrightarrow$};

\draw[name path=A] (4,2.5) .. controls (5,2.5) .. (5,4) ;
\draw[name path=B] (5.5,0) .. controls (5.5,1.5) .. (6.5,1.5);
\draw[white, name path=C] (4,0) -- (5.5,0);
\draw[white, name path=D] (5,4) -- (6.5,4);
\draw[thick, name path=E] (4.9,2.7)--(6.5,3.2);
\tikzfillbetween[of=D and B]{blue, opacity=0.4};
\tikzfillbetween[of=A and C]{blue, opacity=0.4};
\filldraw[white, opacity=0.4] (4,0)--(6.5,0)--(6.5,3.2)--(4.9,2.7)--(4,2.7)--(4,0);
\node at (5.9,3.4) {$N_2$};
\draw[->-] (5.5,0.5)--(5.5,0.6) {};
\draw[->-] (4.95,3)--(4.95,2.9) {};

\node at (5,4.5) {$\updownarrow$};

\node at (7.5,2) {$\not\leftrightarrow$};

\draw (4.5,5)--(4.5,9);
\draw (4.5,7)--(6,7);
\filldraw[blue,opacity=0.4] (4.5,7)--(6,7)--(6,9)--(4.5,9);
\filldraw[blue,opacity=0.2] (4.5,7)--(6,7)--(6,5)--(4.5,5);
\draw[->-] (4.5,8)--(4.5,7.9) {};

\draw[name path=A] (8,2.5) .. controls (9,2.5) .. (9,4) ;
\draw[name path=B] (9.5,0) .. controls (9.5,1.5) .. (10.5,1.5);
\draw[white, name path=C] (8,0) -- (9.5,0);
\draw[white, name path=D] (9,4) -- (10.5,4);
\draw[name path=E] (8,0) -- (10.5,4);
%\draw[thick, name path=E] (4.9,2.7)--(6.5,3.2);
\tikzfillbetween[of=D and B]{blue, opacity=0.4};
\tikzfillbetween[of=A and C]{blue, opacity=0.4};
\filldraw[white, opacity=0.4] (8,0)--(10.5,0)--(10.5,4);
\node at (9.8,3.5) {$N_3$};
\draw[->-] (9.5,0.5)--(9.5,0.6) {};
\draw[->-] (8.95,3)--(8.95,2.9) {};

\node at (9,4.5) {$\updownarrow$};

\filldraw[blue,opacity=0.4] (8,5)--(10.5,5)--(10.5,9)--(8,9);
\filldraw[white,opacity=0.4] (8,5)--(10.5,5)--(10.5,9);
\draw (8,5)--(10.5,9);

\end{tikzpicture}
\caption{Lower row of diagrams: the shaded domain $W\subset\RR^2$ bounded by the two components $Y_1, Y_2$ of a hyperbola, is partitioned in three mutually non-bordant ways according to $N_i, i=1,2,3$. The darkly (resp.\ lightly) shaded region is $W_{+,i}$ (resp.\ $W_{-,i}$). Each lower diagram can be transformed into the one above it, for the purposes of computing $\theta_{W;W_{+,i}}\circ\Exp_W$.
}\label{fig:multiple.boundaries}
\end{center}
\end{figure}
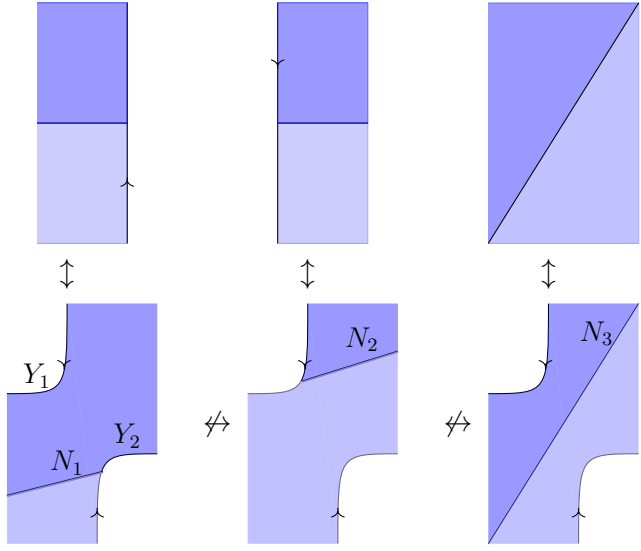

Pick three non-cobordant partitions $N_1, N_2, N_3$ of $W$, as illustrated in the lower row of Fig.\ \ref{fig:multiple.boundaries}. These give rise to homorphisms $\theta_{W;W_{+,i}}:K_1(C^*_W(\partial W))\rightarrow\ZZ$. For $N_2$, we deduce that $\theta_{W;W_{+,2}}\circ\Exp_W([\mathfrak{b}])=1$ by deforming to the standard partition $\cW_+\cup\cW_-$ of the right half-plane $\cW$, and using the same invariance argument as in Eq. \eqref{eqn:edge.index.computation}. Similarly, for $N_1$, we deduce that $\theta_{W;W_{+,1}}\circ\Exp_W([\mathfrak{b}])=-1$ by deforming to a standard partition of the \emph{left} half-plane, and noticing that the latter problem is just a rotation of the standard problem on the right half-plane, with $\cW_+$ and $\cW_-$ exchanged. Finally,  $\theta_{W;W_{+,3}}\circ\Exp_W([\mathfrak{b}])=0$ by deforming to the case where $W^\prime$ is the entire plane, so $\partial W^\prime=\emptyset$ and $\Exp_{W^\prime}=0$.

By observing how $w_{Y_1}$ ``flows'' across $N_i$, it follows immediately that 
\begin{equation*}
\theta_{W;W_{+,1}}[w_{Y_1}]=0=\theta_{W;W_{+,3}}[w_{Y_1}]
\end{equation*}
 and $\theta_{W;W_{+,2}}[w_{Y_1}]=1$. Similarly, 
 \begin{equation*}
 \theta_{W;W_{+,2}}[w_{Y_2}]=0=\theta_{W;W_{+,3}}[w_{Y_2}]
 \end{equation*}
  and $\theta_{W;W_{+,1}}[w_{Y_2}]=-1$, while 
  \begin{equation*}
  \theta_{W;W_{+,1}}[w_Z]=0=\theta_{W;W_{+,2}}[w_Z]
  \end{equation*}
   and $\theta_{W;W_{+,3}}[w_Z]=1$. So we can think of $(\theta_{W;W_{+,1}}, \theta_{W;W_{+,2}}, \theta_{W;W_{+,3}})$ as a surjective $\ZZ$-linear map from the $\ZZ$-submodule of $K_1(C^*_W(\partial W))$ spanned by $[w_{Y_1}], [w_{Y_2}], [w_Z]$ onto the free $\ZZ$-module $\ZZ^3$. Then it follows that $[w_{Y_1}]$, $[w_{Y_2}]$, $[w_Z]$ span $K_1(C^*_W(\partial W))\cong\ZZ^3$. Comparing with $(\theta_{W;W_{+,1}}, \theta_{W;W_{+,2}}, \theta_{W;W_{+,3}})$ applied to $\Exp_W([\mathfrak{b}])$, we deduce that $\Exp_W([\mathfrak{b}])=[w_{Y_1}]+[w_{Y_2}]$ is represented by the sum of edge-travelling operators along each boundary component. We conclude that the gap-filling phenomenon persists, for Chern insulator Hamiltonians $H_{{\rm Chern}, W}$ acting on this domain $W$ with more than one boundary component.

\section{Quantised boundary currents}\label{sec:cyclic.cocycles}

Recall the conventions for $W\subset X$ from \S \ref{sec:gap.filling.phenomenon}: Let $X$ be a complete, connected Riemannian manifold with an effective, cocompact, properly discontinuous, isometric action of a discrete countable group $\Gamma$. Also given is a half-space $W\subset X$ --- a closed subset with measure zero $\partial W$, such that $d(x,X\setminus W)$ is unbounded. 

Let $S$ be the compact separated part of ${\rm Spec}(H_X)$ lying below some resolvent value of $H_X$ (Remark \ref{rem:choice.of.phi}), and let $\Delta$ denote the (bounded) spectral gap of $H_X$ lying immediately above $S$. The spectral projection $P_S$ for $H_X$ can be obtained as $\varphi(H_X)$ with $\varphi\in \mathcal{S}(\RR)$ chosen to be a \emph{Schwartz function}, not just a $C_0(\RR)$ function. Furthermore, we arrange for $\varphi$ to be 1 on $(-\infty,{\rm sup}\, S]\cap {\rm Spec}(H_W)$, not just on $S=(-\infty,{\rm sup}\, S]\cap {\rm Spec}(H_X)$. Then $-\varphi^\prime\in\cS(\RR)$ as a function of ${\rm Spec}(H_W)$ is nonzero only inside $\Delta$, and we further arrange for $-\varphi^\prime$ to be positive (Fig.\ \ref{fig:spectral.projection}). In the previous section, we gave examples of spectral projections $P_S=\varphi(H_X)\in C^*(X,\Gamma)$ and partitions $W_+$ of $W$, such that $[P_S]\mapsto \theta_{W_+}(\Exp_W[P_S])$ is a nontrivial homomorphism indicating that $\Delta\subset {\rm Spec}(H_W)$. In the remaining subsections, we will derive the following \emph{numerical} (i.e.\ not \emph{a priori} quantised) formula for this homomorphism:
\begin{theorem}\label{thm:main.current.theorem}
Assume that $\Gamma$ has polynomial growth. With $P_S=\varphi(H_X)$ a spectral projection as in the above paragraph, and $W_+$ an admissible partition of $W$, 
we have
\begin{equation}
\theta_{W_+}(\Exp_W[P_S])\equiv\theta_{W_+}(\Exp_W[\varphi(H_X)])=-2\pi\, {\rm Tr}(-\varphi^\prime(H_W)\, i[H_{W,\Delta},\Pi]),\label{eqn:current.formula}
\end{equation}
where $H_{W,\Delta}$ denotes the restriction of $H_W$ to its spectral subspace for $\Delta$.
\end{theorem}
Regarding the polynomial growth condition, see the next subsection for details. Our proof combines ideas originating in \cite{Roe-partition} and \cite{KSB,PSB}, as well as some technical results involving smooth integral kernel operators in \cite{Ludewig-Thiang}. 

\medskip

\begin{remark}[{\bf Physical significance of Theorem \ref{thm:main.current.theorem}}]\label{rem:physical.significance} When considering the boundary states of $H_W$ with energies lying in $\Delta$, the term $i[H_{W,\Delta},\Pi]$ is the time-derivative of the observable $\Pi$ of being in $W_+$, by Heisenberg's equation of motion. With $-\varphi^\prime> 0$, we interpret $-\varphi^\prime(H_W)$ as a statistical ensemble of generalised eigenstates of $H_W$ with energies within $\Delta$ (see Fig.\ \ref{fig:spectral.projection}), ``normalised'' by the condition $\int_\Delta -\varphi'=1$. Furthermore, by Theorem \ref{thm:functional.calculus}, we have
$\varpi(-\varphi'(H_W))=-\varphi'(H_X)=0$, so $-\varphi'(H_W)\in C^*_W(\partial W)$ is localised near $\partial W$. Thus, ${\rm Tr}(-\varphi^\prime(H_W)\, i[H_{W,\Delta},\Pi])$ on the right-hand-side of Eq.\ \eqref{eqn:current.formula} is the expected rate of change of probability to be inside $W_+$, in the statistical ensemble $-\varphi^\prime(H_W)$ of boundary localised states. Because this is equal to $\frac{1}{2\pi}$ of some integer Fredholm index by Eq.\ \eqref{eqn:current.formula}, we deduce, \emph{a posteriori}, that the $\Delta$-filling boundary states of $H_W$ constitute a quantised current channel flowing across $N$ from $W_+$ into $W_-$. 
\end{remark}

\begin{example}
Applying Eq. \eqref{eqn:current.formula} to the Landau Hamiltonian, we obtain from Eq.\ \eqref{eqn:edge.index.computation} that $H_{{\rm Lan},W}$ has $j$ quantised edge current channels in the $j$-th spectral gap of $H_{{\rm Lan},X}$ --- this rule-of-thumb is frequently invoked in the physics literature. Our analysis generalises the existing rigorous results obtained for $W=\cW$ a half-plane, such as \S 7.1 of \cite{PSB}, \cite{KSB} and Fig.\ 1\ of \cite{Pule}. 
 \end{example}

\subsection{Subalgebras of smooth kernel operators}

We establish a refinement of \S \ref{sec:functional.calculus}, applicable under the {\em polynomial volume growth} hypothesis on $X$, which means that
  \begin{equation} \label{PolynomialVolumeGrowth}
   V_\mu := \int_X \bigl(1+d(x, y)\bigr)^{-\mu} \dd y < \infty
  \end{equation}
for some $\mu > 0$ and some $x \in X$. By $\Gamma$-invariance, a similar estimate then holds for any $x \in X$. By cocompactness of the action and the Milnor--\v{S}varc Lemma, this is equivalent to requiring that $\Gamma$ be of polynomial volume growth with respect to the word metric. In other words, this condition turns out to be a condition on the group $\Gamma$ alone. Typical examples of groups that satisfy this are crystallographic groups.

\begin{definition}
We define the following subsets of the algebras $C^*(X, \Gamma)$, $C^*_W(\partial W)$ and $Q^*(W, \Gamma)$.
\begin{enumerate}
\item[(1)] A smooth kernel $a\in C^\infty (X\times X)$ has {\em rapid decay away from the diagonal} if for any $\nu \in \RR$, there exists a constant $C_\nu >0$ such that
\begin{equation} \label{RapidDecayEstimate}
 |a(x, y)| \leq C_\nu \bigl(1 + d(x, y)\bigr)^{-\nu}, \qquad x,y,\in X.
\end{equation}
The set of integral operators in $\cB(L^2(X))$ with smooth $\Gamma$-invariant kernels $a \in C^\infty(X \times X)$ with {rapid decay away from the diagonal}, is denoted $\sC(X,\Gamma)$.
\item[(2)] Let $\sC_W(\partial W)$ be the space of integral operators in $\cB(L^2(W))$ with smooth kernels that decay rapidly away from the diagonal, and additionally have the property that for each $\nu \geq 0$, there exists a constant $C_\nu >0$ such that
\begin{equation*}
  |a(x, y)| \leq C_\nu\bigl[\bigl(1+ d_{\partial W}(x)\bigr)^{-\nu} + \bigl(1+ d_{\partial W}(y)\bigr)^{-\nu}\bigr],\qquad x,y,\in W.
\end{equation*}
\item[(3)] $\sQ(W, \Gamma) \subset Q^*(W, \Gamma)$ is the subspace of smooth kernel operators $a$  that can be written as the sum $a = \sigma a^\prime + b$, where $b \in \sC_W(\partial W)$ and $a^\prime \in \sC(X, \Gamma)$, with restriction $\sigma a$ to $W$.
\end{enumerate}
\end{definition}

It is straightforward to show that both $\sC(X, \Gamma) \subset C^*(X, \Gamma)$ and $\sQ(W, \Gamma) \subset Q^*(W, \Gamma)$ are subalgebras, and that the periodization map $\varpi$ takes $\sQ(W, \Gamma)$ to $\sC(X, \Gamma)$. Moreover, the kernel of $\varpi$ restricted to $\sQ(W, \Gamma)$ is precisely $\sC_W(\partial W) \subset C^*_W(\partial W)$ leading to the short exact sequence
\begin{equation} \label{ShortExactSequenceSmooth}
\begin{tikzcd}
  0 \ar[r] & \sC_W(\partial W) \ar[r] & \sQ(W, \Gamma) \ar[r, "\varpi"] & \sC(X, \Gamma) \ar[r] & 0.
\end{tikzcd}
\end{equation}
This refines the quasi-periodic short exact sequence of $C^*$-algebras, Eq.\ \eqref{ShortExactSequence}.

We have the following refinement of Theorem \ref{thm:functional.calculus}.

\begin{proposition}\label{prop:smooth.functional.calculus}
If $\varphi$ lies in the Schwartz space $\cS(\RR)$, we have $\varphi(H_X) \in \sC(X, \Gamma)$ and $\varphi(H_W) \in \sQ(W, \Gamma)$.
\end{proposition}

\begin{proof}
For $\varphi(H_X)$, this is Thm.~6.3 in \cite{Ludewig-Thiang}. The case of $\varphi(H_W)$ can be dealt with in a similar fashion.
\end{proof}

\subsection{Proof of Theorem \ref{thm:main.current.theorem}}\label{sec:quantised.chiral.currents}

With respect to an admissible partition of $W$ into $W_+\cup W_-$, and writing $\Pi$ for $\Pi_{W_+}$ as before, the \emph{switching elements} $(1-\Pi)A\Pi$ and $\Pi A (1-\Pi)$ of an operators $A\in \cB(L^2(W))$ will be of particular interest to us. 
\begin{lemma}\label{lem:switching.elements}
Suppose $A,B\in \cB(L^2(W))$ have trace class switching elements. Then $[\Pi,A]$ (and also $[\Pi,B]$) is trace class with zero trace. Furthermore, 
\begin{equation*}
{\rm Tr}(A[\Pi,B])={\rm Tr}([\Pi A\Pi, \Pi B \Pi]-\Pi[A,B]\Pi)=-{\rm Tr}(B[\Pi,A]).
\end{equation*}
\end{lemma}
\begin{proof}
$[\Pi,A]=\Pi A(1-\Pi)-(1-\Pi)A\Pi$ is the sum of two trace class operators. Since $\Pi^2=\Pi$, it follows from cyclicity of the trace that
\begin{equation*}
{\rm Tr}(\Pi A(1-\Pi))={\rm Tr}(\Pi \Pi A (1-\Pi))={\rm Tr}(\Pi A (1-\Pi)\Pi)={\rm Tr}(0)=0,
\end{equation*}
so ${\rm Tr}([\Pi,A])=0$.
For the second statement, observe that the middle term
\begin{equation*}
[\Pi A\Pi, \Pi B \Pi]-\Pi[A,B]\Pi=-\Pi A(1-\Pi)B\Pi+\Pi B(1-\Pi)A\Pi
\end{equation*}
is trace class. Supplementing it with the terms $-(1-\Pi)A(1-\Pi)B\Pi$ and \mbox{$\Pi B(1-\Pi)A(1-\Pi)$}, which are traceless (by cyclicity), we get
\begin{align*}
{\rm Tr}([\Pi A\Pi, \Pi B \Pi]-\Pi[A,B]\Pi)&= {\rm Tr}(-A(1-\Pi)B\Pi+\Pi B(1-\Pi)A)\\
&={\rm Tr}(-A(1-\Pi)B\Pi+A\Pi B(1-\Pi))\\
&={\rm Tr}(-AB\Pi +A\Pi B\Pi + A\Pi B-A\Pi B\Pi)\\
&={\rm Tr}(A[\Pi,B]).
\end{align*}
\end{proof}

\begin{lemma}\label{lem:trace.class.off.diagonal}
Suppose $u\in C^*_W(\partial W)^+$ is a unitary such that $u$ (or equivalently $u-1$) has trace class switching elements. Then
\begin{equation*}
\theta_{W_+}([u])={\rm Tr}(u[\Pi,u^*])=-{\rm Tr}(u^*[\Pi,u])=-{\rm Tr}((u^*-1)[\Pi,u]).
\end{equation*}
\end{lemma}
\begin{proof} The computation that
\begin{equation*}
\Pi-T_uT_{u^*}=\Pi-\Pi u \Pi u^*=\Pi+\Pi u (1-\Pi)u^*-\Pi uu^*=\Pi u (1-\Pi)u^*
\end{equation*}
is trace class, and similarly for $\Pi-T_{u^*}T_u$, shows that modulo trace class operators, $T_{u*}$ is an inverse for $T_u$ in $\cB(L^2(W_+))$. By Cald\'{e}ron's formula, the Fredholm index of $T_u$ is
\begin{equation*}
\begin{aligned}
\theta_{W_+}([u])\equiv{\rm Index}\, T_u &= {\rm Tr}(T_uT_{u^*}-T_{u^*}T_u)\\
 & = {\rm Tr}(\Pi u \Pi u^*\Pi - \Pi u^* \Pi u \Pi) \\
 &= {\rm Tr}(\Pi u\Pi u^*\Pi-\Pi u^* \Pi u \Pi) \\ 
 &={\rm Tr}([\Pi u \Pi, \Pi u^* \Pi]-\Pi[u,u^*]\Pi)
 %&= {\rm Tr}(\Pi(u-1)\Pi(u^*-1)\Pi-\Pi(u^*-1)\Pi(u-1)\Pi) \\ 
 %&= {\rm Tr}([\Pi (u-1) \Pi, \Pi (u^*-1)\Pi]-\Pi[u-1,u^*-1]\Pi) \\
\end{aligned}
\end{equation*}
The result follows from Lemma \ref{lem:switching.elements} and swapping the roles of $u$ and $u^*$.
\end{proof}

\begin{remark}
The \emph{relative index} of a pair of projections $(P,Q)$ is the integer defined (where possible) by ${\rm dim\, ker}(P-Q-1)-{\rm dim\, ker}(P-Q+1)$, and if $P-Q$ is trace class, the formula ${\rm Ind}(P-Q)={\rm Tr}(P-Q)$ holds, see \cite{ASS} for details. In particular, the expression ${\rm Tr}(u[\Pi,u^*])={\rm Tr}(u\Pi u^*-\Pi)$ in Lemma \ref{lem:trace.class.off.diagonal} computes the relative index of the pair $(u\Pi u^*, \Pi)$, and we have another way to see that ${\rm Tr}(u[\Pi,u^*])$ is in fact integral.
\end{remark}

\begin{proposition}\label{prop:smooth.localised.Roe.switching}
If $A\in \sC_W(\partial W)$, then it has trace class switching elements.
\end{proposition}

\begin{proof}
Let $R_n, n\in\NN$ be a positive sequence increasing to $\infty$, then $W_{n,+}:=W_+\cap Q_{R_n}^{W;W_+}$ (recall (ii) of Definition \ref{defn:admissible.subset}) is an increasing exhaustion of $W_+$ by compact subsets. The projection $\Pi=\Pi_{W_+}$ is strongly approximated by the corresponding sequence $\Pi_{n,+}$ that projects onto $L^2(W_{n,+})$; similarly for $(1-\Pi)=\Pi_{W_-}$. Then the intermediate switching elements $\Pi_{n,+}A\Pi_{n,-}$ is a sequence of trace-class operators (since $A$ has smooth integral kernel), which can be arranged to be Cauchy by an appropriate choices of the $R_n$, due to the rapid decay of the integral kernel away from the diagonal and from $\partial W$, and polynomial volume growth, see Fig.\ \ref{fig:smooth.kernel.small}. Thus the limit $\Pi A(1-\Pi)$ is trace class, and similarly for $(1-\Pi) A\Pi$.
\end{proof}

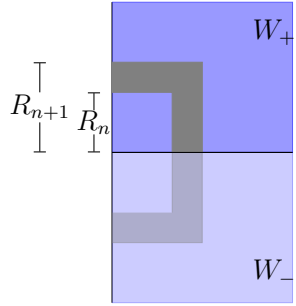
\begin{figure}
\begin{center}
\begin{tikzpicture}[scale=0.8, every node/.style={scale=0.9}]

\draw (0,0)--(0,5);
\draw (0,2.5)--(3,2.5);
\filldraw[blue,opacity=0.4] (0,2.5)--(3,2.5)--(3,5)--(0,5);
\filldraw[blue,opacity=0.2] (0,2.5)--(3,2.5)--(3,0)--(0,0);
\filldraw[gray,opacity=1] (0,3.5)--(1,3.5)--(1,2.5)--(1.5,2.5)--(1.5,4)--(0,4);
\filldraw[gray,opacity=0.4] (0,1.5)--(0,1)--(1.5,1)--(1.5,2.5)--(1,2.5)--(1,1.5);
\draw (0,2.5)--(3,2.5);
%\draw[->-] (2,6)--(2,6.1) {};
\node at (2.7,4.5) {$W_+$};
\node at (2.7,0.5) {$W_-$};

 \draw[|-] (-0.3,2.5) -- (-0.3,2.8);
\draw[-|] (-0.3,3.2) -- (-0.3,3.5);
  \node at (-0.3,3) {$R_n$};
  
 \draw[|-] (-1.2,2.5) -- (-1.2,3);
\draw[-|] (-1.2,3.5) -- (-1.2,4);
  \node at (-1.2,3.25) {$R_{n+1}$};

\end{tikzpicture}
\caption{Let $R_{n+1}>R_n>0$. With the notation of Prop.\ \ref{prop:smooth.localised.Roe.switching}, let $x\in W_{n+1,+}\setminus W_{n,+}$ (dark grey region) and $y\in W_{n+1,-}\setminus W_{n,-}$ (light grey region). If $a\in\sC_W(\partial W)$, then $a(x,y)$ is small because $x$ and $y$ are far apart, or because $x$ and $y$ are far from $\partial W$. 
}\label{fig:smooth.kernel.small}
\end{center}
\end{figure}

\begin{corollary}\label{cor:trace.class}
Let $A\in \sC_W(\partial W)$. If $C\in \cB(L^2(W))$ is such that $AC\in \sC_W(\partial W)$, then $A[\Pi,C]$ is trace class.
\end{corollary}
\begin{proof}
This follows from Prop.\ \ref{prop:smooth.localised.Roe.switching} applied to the equality
\begin{equation*}
A[\Pi,C]=(1-\Pi)A\Pi C+\Pi AC (1-\Pi) - \Pi A(1-\Pi) C - (1-\Pi) AC\Pi.
\end{equation*}
\end{proof}
\begin{example}\label{ex:ideal.trace.class}
If $A\in \sC_W(\partial W)$ and $C\in \sQ(W,\Gamma)$, then $AC\in \sC_W(\partial W)$ by the ideal property, so that $A[\Pi,C]$ is trace class.
\end{example}

We now have all the ingredients needed for the proof of Theorem \ref{thm:main.current.theorem}.

\begin{proof}[Proof of Thm.~\ref{thm:main.current.theorem}]
As in the beginning of this section, let $\varphi(H_X)=\varphi_S(H_X)$ be the spectral projection of $H_X$ with $\varphi\in \mathcal{S}(\RR)$. Then the unitary operator 
\begin{equation*}
W={\rm exp}(-2\pi i \varphi(H_W))
\end{equation*}
 satisfies 
 \begin{equation*}
 \begin{aligned}
 \varpi(W-1) &={\rm exp}\bigl(-2\pi i \cdot \varpi(\varphi(H_W))\bigr)-1\\
 &={\rm exp}\bigl(-2\pi i \cdot\varphi(H_X)\bigr)-1\\
 &=1-1=0
 \end{aligned}
 \end{equation*}
  by Theorem \ref{thm:functional.calculus} and the projection property of $\varphi(H_X)$, so Proposition \ref{prop:smooth.functional.calculus} says that $W-1\in \sC_W(\partial W)$, and similarly for $W^*-1$. Then $(W^*-1)$ has trace class switching elements (Prop.\ \ref{prop:smooth.localised.Roe.switching}), so we can use Lemma \ref{lem:trace.class.off.diagonal} to write a trace formula for the left-hand-side of Eq.\ \eqref{eqn:current.formula},
\begin{equation}\label{eqn:trace.formula.nice.exponential}
\begin{aligned}
\theta_{W_+}({\rm Exp}_W[\varphi(H_X])&\equiv \theta_{W_+}\Bigl(\bigl[{\rm exp}\bigl(-2\pi i \varphi(H_W)\bigr)\bigr]\Bigr)\equiv\theta_{W_+}([W])\\
&= -{\rm Tr}\bigl((W^*-1)[\Pi,W]\bigr).
\end{aligned}
\end{equation}
Expanding $W=e^{-2\pi i \varphi(H_W)}$ as a power series, we have
\begin{equation*}
\begin{aligned}
\theta_{W_+}\bigl(&{\rm Exp}_W[\varphi(H_X)]\bigr)\\
&=-{\rm Tr}\left( (W^*-1)\sum_{k=1}^\infty\frac{(-2\pi i)^k}{k !}\sum_{l=0}^{k-1}(\varphi(H_W))^l[\Pi,\varphi(H_W)](\varphi(H_W))^{k-1-l}\right).%\label{eqn:current.formula.power.series}
\end{aligned}
\end{equation*}
By Corollary \ref{cor:trace.class}, Example \ref{ex:ideal.trace.class} and the fact that $W^*$ commutes with $\varphi(H_W)^l$, $(W^*-1)[\Pi,\varphi(H_W)]$ is trace class, so the above partial sums of operators under the trace are trace class. Using continuity and cyclicity of the trace, we obtain 
\begin{equation*}
\begin{aligned}
\theta_{W_+}({\rm Exp}_W[\varphi(H_X)])&=-{\rm Tr}\left( (W^*-1)\sum_{k=1}^\infty\frac{(-2\pi i)^k}{(k-1)!}(\varphi(H_W))^{k-1}[\Pi,\varphi(H_W)]\right)\\
&= 2\pi i\,{\rm Tr}\bigl((1-W)[\Pi,\varphi(H_W)]  \bigr).
\end{aligned}
\end{equation*}
Let $Q_\Delta$ be the spectral projection of $H_W$ for the interval $\Delta$ (the spectral gap immediately above $S$), and $Q_-$ the spectral projection of $H_W$ for $(-\infty,{\rm sup}\, S]$; note that $Q_\Delta$ and $Q_-$ are orthogonal to each other. Decompose $\varphi(H_W)=(\chi_\Delta\cdot\varphi)(H_W)+ (\chi_{\Delta^c}\cdot\varphi)(H_W)=Q_\Delta\varphi(H_W)+Q_-$. Let us write \mbox{$H_{W,\Delta}=Q_\Delta H_W Q_\Delta$} for the operator $H_W$ restricted to the spectral subspace ${\rm Range}(Q_\Delta)$. Then we can rewrite the decomposition as $\varphi(H_W)=\varphi(H_{W,\Delta})\oplus Q_-$. The $Q_-$ piece will not contribute to ${\rm Tr}((1-W)[\Pi,\varphi(H_W)])$; observe that the function $s\mapsto 1-e^{-2\pi i \varphi(s)}$   vanishes on ${\rm Spec}(H_W)\setminus\Delta$, so $(1-W)Q_-=0$, and thus 
\begin{equation*}
(1-W)[\Pi,Q_-]=(1-W)\Pi Q_-=(1-\Pi)(1-W)\Pi Q_--\Pi(1-W)(1-\Pi)Q_-
\end{equation*}
is trace class by Prop.~\ref{prop:smooth.localised.Roe.switching}, with
\begin{equation*}
\begin{aligned}
{\rm Tr}((1-W)[\Pi,Q_-]) &={\rm Tr}\bigl((1-W)\Pi Q_-\bigr)\\
&={\rm Tr}\bigl((1-W)\Pi Q_-Q_-\bigr)\\
&={\rm Tr}\bigl(\underbrace{Q_-(1-W)}_{0}\Pi Q_-\bigr)=0.
\end{aligned}
\end{equation*}
We obtain the reduction
\begin{equation*}
\theta_{W_+}\bigl({\rm Exp}_W[\varphi(H_X)]\bigr)= 2\pi i\, {\rm Tr}\bigl((1-W)[\Pi,\varphi(H_{W,\Delta})]\bigr).
\end{equation*}
Now, we may regard $\varphi$ as its restriction to $ \overline{\Delta}\supset{\rm Spec}(H_{W,\Delta})$, and approximate it in the $C^1$ sense by a sequence of polynomials $\{\varphi_n\}$ on ${\rm Spec}(H_{W,\Delta})$. Also, pick some $g\in\mathcal{S}(\RR)$ which restricts to the identity function on $\Delta$, to see that
\begin{align*}
(1-W)H_{W,\Delta} &=(1-W)Q_\Delta g(H_W)\\
&=(1-W)g(H_W)\;\in\; \sC_W(\partial W)\cdot\sQ(W,\Gamma)\;\subset\; \sC_W(\partial W)
\end{align*}
is trace class (Corollary \ref{cor:trace.class}). So we may use cyclicity of the trace and the fact that $H_{W, \Delta}$ commutes with $1-W$ to re-sum
\begin{equation*}
\begin{aligned}
{\rm Tr}\Bigl((1-W)[\Pi,\varphi_n(H_{W,\Delta})]\Bigr) &= \sum_{k=0}^\infty a_k {\rm Tr}\Bigl((1-W)[\Pi, H_{W,\Delta}^k]\Bigr)\\
&= \sum_{k=0}^\infty a_k \sum_{l=0}^{k-1}{\rm Tr}\Bigl((1-W)H_{W, \Delta}^k[\Pi, H_{W,\Delta}] H_{W, \Delta}^{k-l-1}\Bigr) \\
&= \sum_{k=0}^\infty a_k {\rm Tr}\Bigl((1-W)H_{W, \Delta}^{k-1}[\Pi, H_{W,\Delta}]\Bigr) \\
&= {\rm Tr}\Bigl((1-W)\varphi_n^\prime(H_{W,\Delta})[\Pi,H_{W,\Delta}]\Bigr),
\end{aligned}
\end{equation*}
where all except finitely many $a_k$ are non-zero, since $\varphi_n$ is a polynomial.
Taking the limit, and noting that $\varphi^\prime(H_{W,\Delta})=\varphi^\prime(H_W)$, 
\begin{equation*}
\theta_{W_+}({\rm Exp}_W[\varphi(H_X)])= 2\pi i\, {\rm Tr}\Bigl((1-W)\varphi^\prime(H_W)[\Pi,H_{W,\Delta}]\Bigr).
\end{equation*}
In fact, for $0\neq k\in\ZZ$, we also have
\begin{equation*}
\begin{aligned}
\theta_{W_+}\bigl({\rm Exp}_W[\varphi(H_X)]\bigr) &=\frac{1}{k}\theta_{W_+}\bigl({\rm Exp}_W[k\varphi(H_X)]\bigr)\\
&= \frac{2\pi i}{k}\, {\rm Tr}\Bigl((1-W^k)k\varphi^\prime(H_W)[\Pi,H_{W,\Delta}]\Bigr)\\
&=2\pi i\, {\rm Tr}\Bigl((1-W^k)\varphi^\prime(H_W)[\Pi,H_{W,\Delta}]\Bigr).
\end{aligned}
\end{equation*}
We claim that $\varphi^\prime(H_W)\in \sC_W(\partial W)$.  First, by Prop.~\ref{prop:smooth.functional.calculus}, we have $\varphi^\prime(H_W) \in \sQ(W, \Gamma)$. But then, since $\varphi^\prime$ is supported in the spectral gap of $H_X$, we have $0=\varphi^\prime(H_X) = \varpi\varphi^\prime(H_W)$ by Thm.~\ref{thm:functional.calculus}.
%\begin{equation*}
%  \varphi^\prime(H_W) = \varphi^\prime(H_W) - \sigma \varphi^\prime(H_X),
%\end{equation*}
%hence $\varphi^\prime(H_W) \in \ker(\varpi)$, by Thm.~\ref{thm:functional.calculus}. 
So $\varphi^\prime(H_W) \in \ker(\varpi)$, and Lemma~\ref{LemmaLocalisedRoeAlgebraKernel} implies that $\varphi^\prime(H_W) \in C^*_W(\partial W) \cap \sQ(W, \Gamma) = \sC_W(\partial W)$, verifying the claim.

The final simplification follows \S 7.1.2 of \cite{PSB}, \S 10 of \cite{KSB}. Let $\phi$ be a smooth function $[0,1]\rightarrow\RR$ vanishing at the endpoints, with Fourier coefficients $a_k, k\in\ZZ$. Observe that $\sum_{k\in\ZZ} a_k =0$, so $a_0=-\sum_{0\neq k\in\ZZ} a_k$. Furthermore, since $\varphi^\prime(H_W)\in \sC_W(\partial W)$, 
$\varphi^\prime(H_W)[\Pi,H_{W,\Delta}]$ is trace class by another application of Corollary~\ref{cor:trace.class}. For brevity, write $\Theta=\theta_{W_+}(\Exp_W[\varphi(H_X)])$, then 
\begin{align*}
a_0\Theta&=-\sum_{0\neq k\in\ZZ} a_k\Theta=-2\pi i \sum_{k\in\ZZ} a_k {\rm Tr}\Bigl((1-W^k)\varphi^\prime(H_{W,\Delta})[\Pi,H_{W,\Delta}]\Bigr)\\
&= -2\pi i\, {\rm Tr}\left(\sum_{k\in\ZZ}a_k e^{-2\pi i k \varphi(H_W)}\varphi^\prime(H_W)[\Pi,H_{W,\Delta}]\right)\\
&=-2\pi i\, {\rm Tr}\Bigl(\phi(\varphi(H_W))\varphi^\prime(H_W)[\Pi,H_{W,\Delta}]\Bigr),\\
&=-2\pi i\, {\rm Tr}\Bigl(\phi(\varphi(H_W))Q_\Delta\varphi^\prime(H_W)[\Pi,H_{W,\Delta}]\Bigr),
\end{align*}
where in the last line, we used $\varphi^\prime(H_W)=Q_\Delta\varphi^\prime(H_W)$. Let $\phi$ converge pointwise and boundedly to the indicator function $\chi_{(0,1)}$, such that the Fourier coefficient $a_0\rightarrow 1$, then $\phi(\varphi(H_W))Q_\Delta\rightarrow Q_\Delta$ in strong operator topology (\cite{RS1} Theorem VIII.5), thus also in weak operator topology. As this is a norm-bounded sequence, it also converges in $\sigma$-weak topology (i.e.\ the weak-$*$-topology on bounded operators regarded as the dual of trace class operators, see Theorem 4.6.14 of \cite{Pedersen}). By continuity of the trace pairing with respect to the $\sigma$-weak topology, the sequence of traces converges,
\begin{equation}
{\rm Tr}\Bigl(\phi\bigl(\varphi(H_W)\bigr)Q_\Delta\varphi^\prime(H_W)[\Pi,H_{W,\Delta}]\Bigr)\overset{\phi\rightarrow\chi_{(0,1)}}{\longrightarrow} {\rm Tr}\Bigl(Q_\Delta\varphi^\prime(H_W)[\Pi,H_{W,\Delta}]\Bigr).
\end{equation}
Since $Q_\Delta \varphi^\prime(H_W) = \varphi^\prime(H_W)$, the desired Eq.\ \eqref{eqn:current.formula} follows.
\end{proof}

\vspace{1em}

\noindent {\bf Acknowledgements.}
The authors thank U.\ Bunke, N.\ Higson and P.\ Hochs for their suggestions, and the Australian Research Council for financial support under grants FL170100020, DE170100149, and DP200100729. G.C.T. acknowledges H. Wang and Q. Wang for their hospitality at East China Normal University where part of this work was done, as well as G.\ De Nittis and E.\ Prodan for helpful discussions.

\bibliography{literature}
\bibliographystyle{plain}

\end{document}